\documentclass[a4paper,11pt]{article}
\usepackage[utf8]{inputenc}
\usepackage{amsmath}
\usepackage{amssymb}
\usepackage{amsthm}
\usepackage{mathtools}
\usepackage{enumerate}
\usepackage[bookmarksnumbered=true,bookmarksopen=true]{hyperref}

\usepackage{parskip}
\makeatletter
\def\thm@space@setup{\thm@preskip=\parskip \thm@postskip=0pt}
\makeatother

\newcommand{\R}{\mathbb{R}}
\newcommand{\C}{\mathbb{C}}
\newcommand{\M}{\mathcal{M}}
\newcommand{\set}[2]{\{\,#1\mid#2\,\}}

\DeclareMathOperator{\real}{Re}
\DeclareMathOperator{\imag}{Im}
\DeclareMathOperator{\grad}{grad}
\DeclareMathOperator{\divv}{div}
\DeclareMathOperator{\hess}{Hess}
\DeclareMathOperator{\dvol}{dvol}
\DeclareMathOperator{\second}{II}
\DeclareMathOperator{\tr}{tr}

\newtheorem{thm}{Theorem}
\newtheorem{prop}{Proposition}
\newtheorem{lem}{Lemma}
\theoremstyle{definition}
\newtheorem{rem}{Remark}

\begin{document}

\title{Quantum dynamics of a particle constrained to lie on a surface}
\author{Gustavo de Oliveira\footnote{Institute for Applied Mathematics, University of Bonn, Endenicher Allee 60, 53115 Bonn, Germany. E-mail address: gustavo.de.oliveira@hcm.uni-bonn.de.} \thanks{Supported by ERC Grant MAQD 240518 (and previously by FAPESP fellowship no. 2010/06971-9 at Universidade Federal de S\~{a}o Carlos, where part of this work was done).}}
\date{\normalsize August 29, 2014}
\maketitle

\begin{abstract}
  \noindent We consider the quantum dynamics of a charged particle in Euclidean space subjected to electric and magnetic fields under the presence of a potential that forces the particle to stay close to a compact surface. We prove that, as the strength of this constraining potential tends to infinity, the motion of this particle converges to a motion generated by a Hamiltonian over the surface superimposed by an oscillatory motion in the normal directions. Our result extends previous results by allowing magnetic potentials and more general constraining potentials.
\end{abstract}

\section{Introduction}

The time evolution of the wave function $\psi$ for a non-relativistic charged particle without spin in Euclidean space subjected to electric and magnetic fields is determined by the Schr\"{o}dinger equation together with an initial condition: $i \partial_t \psi = H \psi$ with $\psi|_{t=0} = \psi_0$, where $\psi_0 \in L^2(\R^3)$ and $H$ is the Hamiltonian operator which acts on $L^2(\R^3)$ as
\[
  H \psi = i \divv(i \grad(\psi) + A \psi) + \langle A, i \grad(\psi) + A \psi \rangle + V \psi.
\]
Here, the inner product is Euclidean and the gradient and divergence are calculated using the Euclidean metric. The magnetic vector potential $A$ and the electric scalar potential $V$ are defined on $\R^3$. We will consider the motion of this particle when its position is constrained to lie on a smooth compact $2$-submanifold $\Sigma \subset \R^3$. Our goal here is to prove that this motion can be effectively described by a unitary group on $L^2(\Sigma)$.

A candidate for providing this effective description is the following: The time evolution of the wave function for the particle on $\Sigma$ is generated by the Hamiltonian $H$ as above, but now $L^2(\R^3)$ is replaced by $L^2(\Sigma)$, the inner product on the tangent spaces of $\Sigma$ is inherited from $\R^3$, the gradient and divergence are calculated using the inherited metric, the vector field $A$ is tangent to $\Sigma$, and $V$ is a function on $\Sigma$. The description of motion obtained in this way is \emph{intrinsic} to $\Sigma$ because it depends only on the inherited metric. 

In contrast to the above description of unconstrained motion (which is based on a fundamental equation), this description of constrained motion is a mere idealization. Here is one way of justifying it: Instead of considering a particle moving on $\Sigma$, we can consider a particle moving in $\R^3$ subjected to a potential that forces the particle to stay near $\Sigma$. The motion of the particle is then generated by the Hamiltonian
\[
  H_\lambda = H + \lambda^4 W
\]
on $L^2(\R^3)$ for large $\lambda$, where $W$ is a positive potential on $\R^3$ that vanishes on $\Sigma$ (we put the power $4$ on $\lambda$ to simplify the notation later). Thus, to justify the above description of constrained motion, we need to answer the following question: As $\lambda$ tends to infinity, does the motion generated by $H_\lambda$ converges to the intrinsic constrained motion?

The answer to this question is no, but if we slightly modified it, the answer is yes. Namely, as $\lambda$ tends to infinity, the motion generated by $H_\lambda$ converges to the motion generated by the constrained Hamiltonian with an additional scalar potential (which depends on the Gaussian and mean curvatures of $\Sigma$) superimposed by an oscillatory motion in the normal directions to $\Sigma$. Since the mean curvature is not an intrinsic quantity, this limit motion is \emph{not intrinsic} to $\Sigma$. This is known when the magnetic potential $A$ is zero and the constraining potential $W$ is exactly quadratic in the normal directions to $\Sigma$ \cite{FH}. Our goal here is to extend this result to the case of nonzero $A$ and more general $W$ in the present setting (in \cite{FH}, the dimension of the ambient space and the codimension of the submanifold are arbitrary). Our hypotheses on the potentials are the following. Roughly speaking, we assume that $W$ has the shape of a well in the normal directions to $\Sigma$ and that the derivative of $W$ in the parallel directions to $\Sigma$ are small near $\Sigma$. There are no restrictions for the electric and magnetic potentials (see hypotheses in Theorem \ref{t:main}).

Models for constrained dynamics under the presence of magnetic field have been considered recently in the physics literature (see \cite{FC,ADA} and references therein). Our paper gives a full mathematical justification for the equations written in \cite{FC}. Other thing that we do here is to prove that the Hamiltonian that generates the limit motion does not depend on the component of $A$ that is normal to $\Sigma$. This settles an issue that was raised in \cite{E} and clarified in \cite{FC}.

In the absence of magnetic field the problem of constraints that we described above was considered previously by Jensen and Koppe \cite{JK}, da Costa \cite{daC1,daC2}, Tolar \cite{T}, and more recently by Mitchell \cite{M}, who provided the most general formal expansions (see also Dell'Antonio and Tenuta \cite{DT} for another approach). The first mathematical analysis appeared in Froese and Herbst \cite{FH}. This is the main reference for our work. Later, Wachsmuth and Teufel \cite{WT} considered the problem from a different point of view (using a different scaling). The lists of references in \cite{FH,M,WT} contain most of the related work to date.

Concerning the mathematical analysis including magnetic field, we are aware of the recent work of Krej\v{c}i\v{r}\'{i}k, Raymond and Tu\v{s}ek \cite{KRT}, who considered the Dirichlet problem for the magnetic Laplacian in a setting similar to ours. Here we consider the dynamical problem in a certain limit starting from the Schr\"odinger equation in Euclidean space.

Our main result (Theorem \ref{t:main}) appears in Section \ref{s:main} after we introduce some notation and set the problem is a suitable way. This section also contains an outline of the method of proof. Next, in Section \ref{s:coord}, we write the local coordinate expressions for all the quantities that we need for the proofs. The Proposition \ref{p:transfer} that we prove in Section \ref{s:transfer} stands together with Theorem \ref{t:main} to support the analysis outlined in Section \ref{s:main}. They are independent statements. In Sections \ref{s:ebounds} and \ref{s:pbounds}, we prove some lemmas which we combine in Section \ref{s:proof} to prove Theorem \ref{t:main}.

\section{Main result and method of proof} \label{s:main}

We will first set up the problem in a suitable way and then we will state our main result. We follow the method of Froese and Herbst \cite{FH}.

Since we will reduce our analysis to the study of the dynamics near $\Sigma$, and we want to identify parallel and normal motions to $\Sigma$, it is convenient to perform a change of variables which moves our problem from a tubular neighborhood of $\Sigma$ in $\R^3$ to the normal bundle of $\Sigma$.

The normal bundle of $\Sigma$ is the submanifold of $T \R^3$ defined by
\[
  N \Sigma = \set{(\sigma, n) \in T \R^3}{\sigma \in \Sigma \ \text{and} \ n \in N_\sigma \Sigma},
\]
where $N_\sigma \Sigma$ denotes the normal space to $\Sigma$ at $\sigma$, and $T\R^3$ denotes the tangent bundle of $\R^3$. Note that $T \R^3$ is diffeomorphic to $\R^3 \times \R^3$ (in a trivial way). For a constant $c > 0$, we define
\[
  N \Sigma_c = \set{(\sigma, n) \in N \Sigma}{|n| < c}.
\]
Here $|\cdot|$ denotes the Euclidean norm. Define a map $E : N \Sigma \to \R^3$ by
\[
  E(\sigma, n) = \sigma + n.
\]
Since $\Sigma$ is compact, by the Tubular Neighborhood Theorem there exists a constant $\delta > 0$ such that this map is a diffeomorphism from $N \Sigma_\delta$ onto a tubular neighborhood of $\Sigma$ in $\R^3$ \cite[Theorem 6.24]{L}. Using the map $E$, we can pull back the Euclidean metric from $\R^3$ to $N \Sigma_\delta$. This provides a metric on $N \Sigma_\delta$. More generally, we can pull back forms and push forward vector fields from a tubular neighborhood of $\Sigma$ to $N \Sigma_\delta$.

We want to study the time evolution of $\psi_{0,\lambda} \in L^2(\R^3)$ generated by $H_\lambda$ in the case where $\psi_{0,\lambda}$ is supported near $\Sigma$. We will show that, up to an error term that vanishes as $\lambda$ tends to infinity, the time evolution of such initial condition stays near $\Sigma$ (Proposition \ref{p:transfer}, Section \ref{s:transfer}). Thus, we can impose Dirichlet boundary condition on the boundary of a tubular neighborhood of $\Sigma$, and then we can move our problem to $L^2(N \Sigma_\delta, \dvol)$. Here, $\dvol$ is calculated using the pulled back metric. By extending the pulled back metric and the Hamiltonian $H_\lambda$ to the complement of $N \Sigma_\delta$, we can consider our problem in $N\Sigma$ without boundary condition. Again, this can be done up to an error term that vanishes as $\lambda$ tends to infinity (Proposition \ref{p:transfer}). Therefore, we can consider the Hamiltonian $H_\lambda$ acting on $L^2(N \Sigma, \dvol)$, where $\dvol$ is calculated using the extended metric.

More precisely, suppose that $g_{N \Sigma}$ is a complete Riemannian metric on $N \Sigma$ that equals the pulled back metric on $N \Sigma_\delta$. Denote by $\dvol$ the Riemannian density for $g_{N \Sigma}$. Let $V$ and $W$ be smooth functions on $N \Sigma$ that equal the corresponding pulled back potentials on $N \Sigma_\delta$. Suppose that $V$ is bounded. Let $A$ be a smooth bounded vector field on $N \Sigma$ that equals the corresponding pushed forward vector potential on $N \Sigma_\delta$. We will study the time evolution generated by $H_\lambda$ acting on $L^2(N \Sigma, \dvol)$ as
\begin{equation} \label{Hlam}
  \boxed{
  H_\lambda \psi = i \divv(i \grad(\psi) + A \psi) + \langle A, i \grad(\psi) + A \psi \rangle + V \psi + \lambda^4 W \psi.
  }
\end{equation}

We next explain how to decompose vectors in the tangent and cotangent spaces of $N \Sigma$ into horizontal and vertical components. This is useful in our analysis in order to identify the parallel and normal motions to $\Sigma$.

Let $\pi : N \Sigma \to \Sigma$ be the projection
\[
  \pi(\sigma, n) = \sigma.
\]
We define the vertical subspace of $T_{\sigma, n} N \Sigma$ as the kernel of the mapping
\[
  d \pi_{\sigma,n} : T_{\sigma,n} N \Sigma \to T_\sigma \Sigma,
\]
where $d\pi_{\sigma,n}$ is the derivative of $\pi$ at $(\sigma,n)$, and $T_{\sigma, n} N \Sigma$ and $T_\sigma \Sigma$ denote the tangent space to $N \Sigma$ at $(\sigma, n)$ and the tangent space to $\Sigma$ at $\sigma$, identified with subspaces of $\R^3 \times \R^3$ and $\R^3$, by means of the usual identifications, respectively. We then define the horizontal subspace of $T_{\sigma, n} \, N \Sigma$ as the orthogonal complement of the vertical subspace in $T_{\sigma, n} N \Sigma$ with respect to the inner product defined by the metric $g_{N \Sigma}$. We will denote by $P_H$ and $P_V$ the projections of the tangent spaces of $N \Sigma$ onto the horizontal and vertical subspaces. In view of the identification of $T_{\sigma,n} N \Sigma$ with $T^*_{\sigma, n} N \Sigma$ given by the inner product, we obtain a decomposition of cotangent vectors into horizontal and vertical components as well. We will denote by $P^H$ and $P^V$ the projections of the cotangent spaces of $N \Sigma$ onto the horizontal and vertical subspaces.

Since the restriction of $d\pi_{\sigma,n}$ to the horizontal subspace of $T_{\sigma,n} N \Sigma$ is an isomorphism, the adjoint $d\pi_{\sigma,n}^* : w \mapsto w \circ d\pi_{\sigma,n}$ is an isomorphism of $T_\sigma^* \Sigma$ onto the horizontal subspace of $T_{\sigma,n}^* N \Sigma$. Thus
\begin{equation} \label{J}
  J_{\sigma,n} \coloneqq (d\pi_{\sigma,n}^*)^{-1}
\end{equation}
is well-defined on the horizontal subspace of $T_{\sigma,n}^* N \Sigma$. The map $J_{\sigma,n}$ identifies the horizontal subspace of $T_{\sigma,n}^* N \Sigma$ with $T_\sigma^* \Sigma$.

More concretely, the decomposition on $N \Sigma_\delta$ can be understood as follows. For each $\sigma \in \Sigma$, we have the decomposition $T_\sigma \R^3 = T_\sigma \Sigma \oplus N_\sigma \Sigma$. Using the identification of $T_\sigma \R^3$ with $\R^3$, and the identification of $T_\sigma \Sigma$ and $N_\sigma \Sigma$ with subspaces of $\R^3$, we obtain a decomposition of $\R^3$ for each point $\sigma$. We will denote by $P_\sigma^T$ and $P_\sigma^N$ the orthogonal projections onto the tangent and normal subspaces. Since we are considering $N \Sigma_\delta$ as a $3$-submanifold of $T \R^3$, we can identify $T_{\sigma,n} N \Sigma_\delta$ with the $3$-dimensional subspace of $\R^3 \times \R^3$ that consists of all the vectors of the form $(X,Y) = \frac{d}{dt} (\sigma(t), n(t)) \big|_{t=0}$, where $(\sigma(t), n(t))$ is a path in $N \Sigma_\delta$ that crosses $(\sigma, n)$ at $t=0$. We define the inner product of two such tangent vectors by
\begin{equation} \label{ip}
  \begin{split}
    \langle (X_1,Y_1), (X_2, Y_2) \rangle & \coloneqq g_{N \Sigma}((X_1, Y_1), (X_2, Y_2)) \\
    & = \langle dE_{\sigma,n}(X_1,Y_1), dE_{\sigma,n}(X_2,Y_2) \rangle\\
    & = \langle X_1 + Y_1, X_2 + Y_2 \rangle,
  \end{split}
\end{equation}
where $dE_{\sigma,n} : \, T_{\sigma,n} N \Sigma_\delta \to T_{\sigma+n} \R^3$ is the derivative of $E$ at $(\sigma,n)$ with the usual identifications, and the inner product on the right-hand side is Euclidean. Thus, the decomposition of $(X, Y) \in T_{\sigma,n} N \Sigma_\delta$ into horizontal and vertical components is
\[
  (X, Y) = (X, P_\sigma^T Y) + (0, P_\sigma^N Y).
\]
This induces a decomposition for cotangent vectors in $T^*_{\sigma, n} N \Sigma_\delta$.

To prove that the limit motion does not depend on the component of $A$ that is normal to $\Sigma$, we introduce a gauge transformation.

Let $\gamma$ be a real-valued function on $N \Sigma$. Then the transformation $S_\gamma$ defined to be the multiplication operator by $e^{i\gamma}$ on $L^2(N\Sigma,\dvol)$ is unitary. Thus the time evolution generated by $S_\gamma^* H_\lambda^{} S_\gamma^{}$ is unitarily equivalent to the time evolution generated by $H_\lambda$. The operator $S_\gamma^* H_\lambda^{} S_\gamma^{}$ is given by the expression in \eqref{Hlam} with $A$ replaced by $A - \grad(\gamma)$. We will consider the time evolution generated by $S_\gamma^* H_\lambda^{} S_\gamma^{}$ with $\gamma$ such that
\begin{equation} \label{gauge}
  \boxed{
  P_V (A - \grad(\gamma)) = 0 \ \text{on} \ N\Sigma_\delta \qquad \text{and} \qquad P_H \grad(\gamma)|_{\Sigma}^{} = 0.
  }
\end{equation}
In Section \ref{s:H}, we will prove that this condition can be fulfilled.

To identify the limit motion we will consider a sequence of orbits generated by $H_\lambda$ with initial conditions whose support is squeezed towards $\Sigma$ as $\lambda$ tends to infinity. This will be implemented by using a dilation transformation in the normal directions. Equivalently, we will consider a sequence of orbits generated by $H_\lambda$ conjugated by unitary dilations, but with a fixed initial condition.

For $\lambda > 0$, define a dilation (in the normal direction) $d_\lambda : N \Sigma \to N \Sigma$ by
\[
  d_\lambda(\sigma,n) = (\sigma,\lambda n),
\]
and set
\[
  \dvol_\lambda = \lambda (d_\lambda^{-1})^* \dvol,
\]
where $(d_\lambda^{-1})^*$ denotes the pull-back by $d_\lambda^{-1}$ (see \eqref{dens} for the local coordinates expressions). Then the transformation $D_\lambda : L^2(N \Sigma, \dvol_\lambda) \to L^2(N \Sigma, \dvol)$ defined by
\[
  D_\lambda \psi = \sqrt{\lambda} \ \psi \circ d_\lambda
\]
is unitary. Note that the space $L^2(N \Sigma, \dvol_\lambda)$ depends on $\lambda$. To obtain a fixed space as $\lambda$ tends to infinity, we introduce another transformation. Set
\[
  \dvol_{N \Sigma} = \lim_{\lambda \to \infty} \dvol_\lambda,
\]
and consider the quotient of densities $\dvol_{N \Sigma}/\dvol_\lambda$, which is a function on $N \Sigma$. Then the transformation $M_\lambda : L^2(N \Sigma, \dvol_{N \Sigma}) \to L^2(N \Sigma, \dvol_\lambda)$ defined by
\[
  M_\lambda \psi = \sqrt{\dvol_{N \Sigma}/\dvol_\lambda} \ \psi
\]
is unitary. Thus the transformation $U_\lambda : L^2(N \Sigma, \dvol_{N \Sigma}) \to L^2(N \Sigma, \dvol)$ defined by
\[
  U_\lambda = D_\lambda M_\lambda
\]
is unitary. Note that, as $\lambda$ tends to infinity, the support of the function $S_\gamma U_\lambda \psi_0$ is squeezed towards $\Sigma$. This is easy to verify using the expressions in Section \ref{s:H}. Therefore, we will consider the time evolution generated by the conjugated Hamiltonian
\[
  L_\lambda^{} \coloneqq U_\lambda^* S_\gamma^* H_\lambda^{} S_\gamma U_\lambda^{}
\] 
with initial condition $\psi_0 \in L^2(N \Sigma,\dvol_{N \Sigma})$. This evolution is unitarily equivalent to the time evolution generated by $H_\lambda$ with initial condition $S_\gamma U_\lambda \psi_0 \in L^2(N \Sigma, \dvol)$.

The next step in our analysis is performing a large $\lambda$ expansion in $L_\lambda$. With the gauge condition in \eqref{gauge}, this formally yields
\[
  \boxed{
  L_\lambda = H_\Sigma + \lambda^2 H_{O,\lambda} + O(\lambda^{-1}),
  }
\]
where $H_\Sigma$ and $H_{O,\lambda}$ are defined by the quadratic forms
\[
  \langle \psi, H_\Sigma \psi \rangle = \int_{N \Sigma} \big( | J P^H (d + A(\sigma,0)) \psi |^2_\sigma + (V(\sigma,0) + K(\sigma)) |\psi|^2 \big) \dvol_{N \Sigma}
\]
and
\[
  \langle \psi, H_{O,\lambda} \psi \rangle = \int_{N \Sigma} \big( | P^V d \psi |^2_{\sigma,n} + (\tfrac{1}{2} \langle n, B(\sigma) n \rangle + F_\lambda(\sigma,n)) |\psi|^2 \big) \dvol_{N \Sigma}
\]
with
\begin{align*}
  K & = s - h^2, \\
  B(\sigma) & = (\hess_n W)(\sigma,0), \\
  F_\lambda(\sigma,n) & = \lambda^{-1} \sum_{|\alpha|=3} a_\alpha(\sigma) n^\alpha + \lambda^{-2} \sum_{|\beta|=4} b_\beta(\sigma) n^\beta.
\end{align*}
(The local coordinate expressions for $H_\Sigma$ and $H_{O,\lambda}$ are given in \eqref{Hs}.) The scalar potential in $H_{O,\lambda}$ comes from a Taylor series expansion (of fifth order in $n$) for $W$ about $(\sigma,0)$ (the remainder term is contained in $O(\lambda^{-1})$ in $L_\lambda$). Here, $\alpha$ and $\beta$ are multi-indices, $|\cdot|_{(\cdot)}^2 = \langle \cdot, \cdot \rangle_{(\cdot)}$, $d$ is the differential operator, $s$ and $h$ are the Gaussian and mean curvatures of $\Sigma$, respectively, and $\hess_n$ denotes the Hessian with respect to the normal variables. As mentioned earlier, the Hamiltonian $H_\Sigma$ is what we would expect to be the generator of a motion on $\Sigma$, but with an additional scalar potential $K$ that depends on the embedding of $\Sigma$ into $\R^3$ because the mean curvature does. Note that $H_\Sigma$ does not depend on the component of $A$ that is normal to $\Sigma$. This component is actually contained in $O(\lambda^{-1})$ in $L_\lambda$. The Hamiltonian $H_{O,\lambda}$ represents a quantum harmonic oscillator in the normal directions with a perturbative potential $F_\lambda(\sigma,n) = O(\lambda^{-1})$. Note that we can not absorb $F_\lambda$ into $O(\lambda^{-1})$ in $L_\lambda$ because $\lambda^2 F_\lambda(\sigma,n) = O(\lambda)$.

We can interpret the Hamiltonians $H_\Sigma$ and $H_{O,\lambda}$ better if we introduce another another metric on $N \Sigma$. For tangent vectors $(X_1,Y_1)$ and $(X_2,Y_2)$ in $T_{\sigma,n} N \Sigma$, we define
\[
  \langle (X_1,Y_1), (X_2,Y_2) \rangle_\lambda = \langle X_1, X_2 \rangle + \lambda^{-2} \langle P_\sigma^N Y_1, P_\sigma^N Y_2 \rangle,
\]
where the inner product on the right-hand side is Euclidean. (In Section \ref{s:H}, we will describe this inner product as a limit of the scaled pulled back inner product.) Thus we can write
\begin{equation} \label{Hsum}
  \begin{split}
    (H_\Sigma + \lambda^2 H_{O,\lambda}) \psi & = i \divv_\lambda( i \grad_\lambda(\psi) + A_\Sigma \psi) + \langle A_\Sigma, i \grad_\lambda(\psi) + A_\Sigma \psi \rangle_\lambda \\
                                              & \quad + (\tfrac{1}{2} \langle n, B(\sigma) n \rangle + F_\lambda(\sigma,n) + V(\sigma,0) + K(\sigma)) \psi,
  \end{split}
\end{equation}
where $A_\Sigma(\sigma) = P^H A(\sigma,0)$, and the gradient and divergence are calculated using the metric defined by $\langle \cdot, \cdot \rangle_\lambda$.

Since the metric defined by $\langle \cdot, \cdot \rangle_\lambda$ is complete, the operator $H_\Sigma + \lambda^2 H_{O,\lambda}$ is essentially self-adjoint on $C_0^\infty(N\Sigma)$ by \cite[Theorem 5.2]{S}. We can not apply this theorem directly to $H_\Sigma$ and $H_{O,\lambda}$, but by adapting its proof we can show that these operators are essentially self-adjoint on $C_0^\infty(N\Sigma)$. We do not give the details here.

We can now state our main result.

\begin{thm} \label{t:main}
  Consider a smooth compact $2$-submanifold $\Sigma \subset \R^3$. Let $g_{N\Sigma}$ be a complete Riemannian metric on $N \Sigma$ that equals the pulled back metric on $N \Sigma_\delta$. Suppose that $A$, $V$ and $W$ are $C^\infty$-functions on $N\Sigma$ that equal the corresponding pulled back and pushed forward potentials on $N\Sigma_\delta$. Suppose that $A$ and $V$ are bounded and suppose that for all $(\sigma,n) \in N \Sigma$ we have:
  \begin{enumerate}[\rm (i)]
    \item
      $W(\sigma,n) \ge \kappa |n|^2$ for some constant $\kappa > 0$.
    \item
      $W(\sigma,0)=0$, $dW(\sigma,0)=0$, and $\hess W(\sigma,0) > 0$.
    \item \label{hyp}
      $\partial_\sigma W(\sigma,n) = O(|n|^5)$ as $|n| \to 0$ for each $\sigma$, where $\partial_\sigma$ denotes the derivative with respect to $\sigma$.
  \end{enumerate}
  Consider the Hamiltonian $H_\lambda$ given in \eqref{Hlam} acting on $L^2(N\Sigma, \dvol)$, and the Hamiltonian $L_\lambda = U_\lambda^* S_\gamma^* H_\lambda S_\gamma U_\lambda^{}$ acting on $L^2(N \Sigma, \dvol_{N\Sigma})$ with $U_\lambda$ and $S_\gamma$ as above. Then, for every $\psi \in L^2(N\Sigma, \dvol_{N\Sigma})$, $T > 0$, and $0 < s < 1$, there exist constants $\lambda_0$ and $C$ such that, for every $\lambda > \lambda_0$, we have
  \begin{equation} \label{main}
    \sup_{t \in [0,T]} \big\| e^{-itL_\lambda} \psi - e^{-it (H_\Sigma + \lambda^2 H_{O, \lambda})} \psi \big\| \le C \lambda^{s-1}.
  \end{equation}
\end{thm}

Theorem \ref{t:main} says that the left-hand side of \eqref{main} tends to zero as $\lambda$ tends to infinity. We remark that neither orbits in \eqref{main} converge as $\lambda$ tends to infinity. It is only their difference that converges. We remark also that Theorem \ref{t:main} combined with Proposition \ref{p:transfer} yield a statement about convergence of wave functions on Euclidean space. This relates the evolution determined by the Schr\"odinger equation to the evolution generated by $H_\Sigma + \lambda^2 H_{O,\lambda}$.

Since $\Sigma$ has codimension one, the normal bundle of $\Sigma$ is trivial, and thus $L^2(N\Sigma, \dvol_{N\Sigma}) = L^2(\Sigma, \dvol_\Sigma) \otimes L^2(\R,dy)$. Furthermore $[H_\Sigma, H_{O,\lambda}] = 0$ because of hypothesis (iii) in Theorem \ref{t:main}. Therefore we have $H_\Sigma = h_\Sigma \otimes I$ for a Hamiltonian $h_\Sigma$ acting on $L^2(\Sigma,\dvol_\Sigma)$, and $H_{O,\lambda} = I \otimes h_{O,\lambda}$ for a Hamiltonian $h_{O,\lambda}$ acting $L^2(\R)$. Consequently
\[
  \begin{split}
    \exp(-it(H_\Sigma + \lambda^2 H_{O,\lambda})) & = \exp(-it H_\Sigma) \exp(-it \lambda^2 H_{O,\lambda}) \\
                                                  & = \exp(-ith_\Sigma) \otimes \exp(-\lambda^2 it h_{O,\lambda}).
  \end{split}
\]
This can be interpreted as a motion on $\Sigma$, described by a unitary group on $L^2(\Sigma,\dvol_\Sigma)$, superimposed by normal oscillations to $\Sigma$ (which can be effectively ignored). Thus Theorem \ref{t:main} says that the motion generated by $L_\lambda$, which is unitarily equivalent to the motion generated by $H_\lambda$, converges to a superposition of motion on $\Sigma$ and normal oscillations as $\lambda$ tends to infinity.

Theorem \ref{t:main} allows for a constraining potential which is not constant in the parallel directions to $\Sigma$ (for example $W(\sigma,n) = |n|^2 + |n|^4 + f(\sigma) |n|^6$, where $f$ is a function).

The relation between the potential $\lambda^4 W$ and the dilation $d_\lambda$ was chosen so that the normal energy diverges like $\lambda^2$ (the normal energy is the energy associated to $\lambda^2 H_{O,\lambda}$). The reason for this is that the energy levels of an harmonic oscillator with potential $\alpha y^2$ are proportional to $\sqrt{\alpha}$.

\section{Coordinate expressions} \label{s:coord}

In this section we introduce local coordinates for $\Sigma$ and $N \Sigma$ and for its tangent and cotangent bundles. Then we write the quantities that we need for our proofs in local coordinates. When there is no risk of confusion, we will denote functions on a manifold and their coordinate representations by the same letter.

\subsection{Charts and partition of unity}

Let $x(\sigma)$ be a local coordinate map for $\Sigma$, and let $\{ \partial/\partial x_1, \partial/\partial x_2 \}$ and $\{ dx_1, dx_2 \}$ be the standard bases for the tangent and cotangent spaces of $\Sigma$. This yields local coordinates for $T \Sigma$ and $T^* \Sigma$ in the usual way. We will denote by $(x,p) \in \R^2 \times \R^2$ the coordinates of the cotangent vector $\sum_1^2 p_j dx_j$ in the cotangent space at $\sigma(x)$.

Let $\nu(\sigma)$ be a local unit normal vector to $\Sigma$ depending smoothly on $\sigma$. We obtain local coordinates for $N \Sigma$ by defining
\[
  \boxed{
  \begin{aligned}
    x(\sigma,n) & = x(\sigma), \\
    y(\sigma,n) & = \langle \nu(\sigma), n \rangle.
  \end{aligned}
  }
\]
Let $\{ \partial / \partial x_1, \partial/\partial x_2, \partial/ \partial y \}$ and $\{ dx_1, dx_2, dy \}$ be the standard bases for the tangent and cotangent spaces of $N \Sigma$. This yields local coordinates for $T N \Sigma$ and $T^* N \Sigma$ in the usual way. We will denote by $(x,y,p,r) \in (\R^2 \times \R)^2$ the coordinates of the cotangent vector $\sum_1^2 p_j dx_j + r dy$ in the cotangent space at $(\sigma(x), y \nu(\sigma(x)))$.

Let $\mathcal{U}$ be a coordinate domain on $\Sigma$. We will assume that each coordinate domain on $N \Sigma$ and $T^* N \Sigma$ have the form $\set{(\sigma,n)}{\sigma \in \mathcal{U} \ \text{and} \ n \in N_\sigma \Sigma}$ and
\[
  \begin{split}
    \{ \ (\sigma, n, \xi, \eta) \mid \ & \sigma \in \mathcal{U}, \ n \in N_\sigma \Sigma, \\
    & \xi \in T_{\sigma,n}^* N \Sigma \ \text{is horizontal}, \ \eta \in T_{\sigma,n}^* N \Sigma \ \text{is vertical} \ \}.
  \end{split}
\]
Since $\Sigma$ is compact, we can assume that the atlases for $N \Sigma$ and $T^* N \Sigma$ have finitely many charts. Furthermore, there exists $\varepsilon > 0$ such that two points in $N \Sigma$ belong to a single coordinate domain if their projections onto $\Sigma$ are distant from each other by at most $\varepsilon$. The same holds for $T^* N \Sigma$.

Let
\begin{equation} \label{pu}
  \{ \chi_j \} = \{ \chi_j(\sigma) \}
\end{equation}
be a smooth partition of unit subordinated to a regular finite cover $\{ \mathcal{V}_j \}$ of $N \Sigma$. (We can assume that $\chi_j$ depends only on $\sigma$.) This means that each $\chi_j$ is a smooth function on $N\Sigma$ supported in a single coordinate domain $\mathcal{V}_j$ on $N \Sigma$ with $0 \le \chi_j \le 1$ and $\sum_{j=1}^m \chi_j^2 = 1$ with $m < \infty$. (Note the difference with the standard definition.) Observe that $\sum_{j=1}^m \chi_j \grad(\chi_j)=0$.

\subsection{Metrics}

Denote by $\sigma_j(x)$ the vector $\partial \sigma(x)/\partial x_j \in \R^3$ for $j=1,2$. This vector corresponds to the tangent vector $\partial/\partial x_j \in T_\sigma \Sigma$. Thus, the $2 \times 2$ matrix for the metric on $\Sigma$ is
\[
  G_\Sigma = [\langle \sigma_i, \sigma_j \rangle].
\]
(We will denote a matrix with entries $X_{ij}$ by $[X_{ij}]$.)

The tangent vectors $\partial/\partial x_j$ and $\partial/\partial y$ in $T_{\sigma,n} N \Sigma$ correspond to the vectors $(\sigma_j, y d\nu[\sigma_j])$ and $(0,\nu(\sigma))$ in $\R^3 \times \R^3$ for $j=1,2$. Here $\sigma_j = \sigma_j(x)$ and $\nu(\sigma) = \nu(\sigma(x))$. Thus, using the inner product in \eqref{ip}, and observing that $\langle \sigma_j, \nu(\sigma) \rangle = 0$ and $\langle d\nu[\sigma_j], \nu(\sigma) \rangle = 0$, we find that, on $N \Sigma_\delta$,
\[
  \left \langle \frac{\partial}{\partial x_j}, \frac{\partial}{\partial y} \right \rangle = \left \langle \sigma_j + y d\nu[\sigma_j], \nu(\sigma) \right \rangle = 0
\]
and
\[
  \left \langle \frac{\partial}{\partial x_i}, \frac{\partial}{\partial x_j}\right \rangle = \left \langle \sigma_i + y d\nu[\sigma_i], \sigma_j + y d\nu[\sigma_j] \right \rangle = (G_\Sigma)_{ij} + C_{ij}
\]
with
\[
  C_{ij} \coloneqq y ( \langle \sigma_i, d\nu[\sigma_j] \rangle + \langle d\nu[\sigma_i], \sigma_j \rangle ) + y^2 \langle d\nu[\sigma_i], d\nu[\sigma_j] \rangle.
\]
Furthermore,
\[
  \left \langle \frac{\partial}{\partial y}, \frac{\partial}{\partial y} \right \rangle = \langle \nu(\sigma), \nu(\sigma) \rangle = 1.
\]
Therefore, the $3 \times 3$ matrix for the pulled back metric on $N \Sigma_\delta$ is the block diagonal matrix
\[
  \boxed{
  G = \begin{bmatrix} G_\Sigma + C & 0 \\ 0 & 1 \end{bmatrix},
  }
\]
where $C$ is the $2 \times 2$ matrix with entries $C_{ij}$.

To extend the pulled back metric on $N \Sigma_\delta$ to a complete Riemannian metric on $N \Sigma$, we could join the pulled back metric on $N \Sigma_\delta$ to the metric on $N \Sigma \setminus N \Sigma_{2 \delta}$ given in local coordinates by the block diagonal matrix
\[
  \begin{bmatrix} G_\Sigma & 0 \\ 0 & 1 \end{bmatrix}.
\]
The extended metric on $N \Sigma$ would be given by the block diagonal matrix
\[
  \begin{bmatrix} G_\Sigma + \eta C & 0 \\ 0 & 1 \end{bmatrix},
\]
where $\eta = \eta(y)$ is a smooth function on $\R$ that equals $1$ for $|y| \le \delta$ and $0$ for $|y| \ge 2 \delta$ with $0 \le \eta \le 1$. With this particular extension, the local coordinate expressions below remain true on all $N \Sigma$ if $C$ is replaced by $\eta C$. However, this particular extension is not necessary for our proofs. We thus let $G$ be the matrix for any complete Riemannian metric $g_{N \Sigma}$ on $N \Sigma$ that equals the pulled back metric on $N \Sigma_\delta$. (We remark that the above extension is complete.)

We now look at the quantity $G_\Sigma + C$ in more detail. Consider two vector fields $X$ and $Y$ tangent to $\Sigma$, and denote by $dX[Y]$ the directional derivative of $X$ in the direction of $Y$. The second fundamental form of $\Sigma$ is defined by
\[
  \second(X,Y) = \langle \nu, dY[X] \rangle.
\]
Since the Lie-bracket $[X,Y] \coloneqq dY[X] - dX[Y]$ is tangent to $\Sigma$, we have
\[
  \langle \nu, dY[X] \rangle = \langle \nu, [X,Y] \rangle + \langle \nu, dX[Y] \rangle = \langle \nu, dX[Y] \rangle,
\]
and thus the second fundamental form is symmetric in $X$ and $Y$. The Weingarten map $L$ is the symmetric linear transformation on the tangent spaces of $\Sigma$ defined by the equation
\[
  \second(X,Y) = \langle L X, Y \rangle.
\]
Differentiating the identity $0 = \langle \nu, Y \rangle$, we obtain
\[
  0 = d0[X] = d \langle \nu, Y \rangle[X] = \langle d\nu[X], Y \rangle + \langle \nu, dY[X] \rangle.
\]
Thus
\[
  \langle L X, Y \rangle = \second(X,Y) = \langle \nu, dY[X] \rangle = -\langle d\nu[X], Y \rangle
\]
for all tangent vectors $X$ and $Y$. Therefore $L = -d\nu$ and
\[
  \langle d\nu[\sigma_i], \sigma_j \rangle = - \langle L \sigma_i, \sigma_j \rangle = - \langle \sigma_i, L \sigma_j \rangle = \langle \sigma_i, d\nu[\sigma_j] \rangle.
\]
Consequently
\[
  (G_\Sigma)_{ij} + C_{ij} = \langle \sigma_i, \sigma_j \rangle -2y \langle \sigma_i, L \sigma_j \rangle + y^2 \langle L \sigma_i, L \sigma_j \rangle.
\]
A short calculation shows that
\begin{equation} \label{id1}
  \sum_{k=1}^2 \sum_{l=1}^2 (\sigma_k)_i (G_{\Sigma}^{-1})_{kl} (\sigma_l)_j = \delta_{ij},
\end{equation}
where $\delta_{ij}$ is the Kronecker delta. Using this we find
\[
  [\langle L \sigma_i, L \sigma_j \rangle] = [\langle \sigma_i, L \sigma_j \rangle] G_\Sigma^{-1} [\langle \sigma_i, L \sigma_j \rangle].
\]
Therefore
\[
  \begin{split}
    G_\Sigma + C & = G_\Sigma - 2y[\langle \sigma_i, L \sigma_j\rangle ] + y^2 [\langle \sigma_i, L \sigma_j \rangle] G_\Sigma^{-1} [\langle \sigma_i, L \sigma_j \rangle] \\
    & = G_\Sigma (1 - y G_\Sigma^{-1}[\langle \sigma_i, L \sigma_j \rangle])^2.
  \end{split}
\]

We will need the following expansions for our proofs.

\begin{lem}[Expansion formulae for the metric]
  Let $(\mathcal{V}, (x,y))$ be a chart on $N \Sigma$, and let $\|\cdot\|$ be the operator norm (induced by the Euclidean norm). Then for every $(x,y)$ such that $|y| \|G_\Sigma^{-1}[\langle \sigma_i, L \sigma_j \rangle]\| < 1$, we have
  \[
    (G_\Sigma + C)^{-1} = G_\Sigma^{-1} + e_1
  \]
  and
  \[
    \log \left[ \frac{\det(G_\Sigma)}{\det(G_\Sigma+C)} \right]^{1/4} = \frac{1}{2} y \tr(L) + \frac{1}{4} y^2 \tr(L^2) + e_3.
  \]
  Here $e_j$ denotes a smooth (possibly matrix valued) function of $(x,y)$ that vanishes to $j$th order in $y$ at $(x,0)$. (The function $e_j$ behaves like $y^j$ for small $y$.)
  \label{l:metric}
\end{lem}

\begin{proof}
Set $M = G_\Sigma^{-1} [\langle \sigma_i, L \sigma_j \rangle]$ and note that $M$ depends only on $x$. Then calculating we find $\det(1-yM) = 1 + e_1$ and $(1-yM)^{-1} = 1 + e_1$. Hence $(G_\Sigma + C)^{-1} = (1-yM)^{-2} G_\Sigma^{-1} = G_\Sigma^{-1} + e_1$. This proves the first formula. A short calculations using \eqref{id1} shows that $\tr(M)=\tr(L)$. Using this, the identity $\log(\det(X)) = \tr(\log(X))$ (which holds for any positive definite real matrix $X$), and the expansion $\log(1-yM) = -yM - y^2 M^2/2 + e_3$, we obtain the second formula.
\end{proof}

\subsection{Projections}

The metric provides an identification between $T_{\sigma,n} N \Sigma$ and $T_{\sigma,n}^* N \Sigma$ given by the map $Z \mapsto g_{N \Sigma}(Z, \, \cdot \,)$. This map is an isomorphism and has matrix $G$ (with respect to the standard bases). Its inverse has matrix $G^{-1}$.

Denote by $P_V$ and $P_H$ the projections of the tangent spaces of $N \Sigma$ onto the vertical and horizontal subspaces. Observe that, the vertical subspace of $T_{\sigma,n} N \Sigma$ is the span of $\partial/\partial y$. Thus, on $N \Sigma_\delta$, since $\langle \partial/\partial x_j, \partial/\partial y \rangle = 0$ for $j=1,2$, and the projections are orthogonal,
\[
  P_V = \begin{bmatrix} 0 & 0 \\ 0 & 1 \end{bmatrix} \qquad \text{and} \qquad P_H = \begin{bmatrix} I & 0 \\ 0 & 0 \end{bmatrix}.
\]
This induces a decomposition of cotangent vectors. Denote by $P^V$ and $P^H$ the projections of the cotangent spaces of $N \Sigma$ onto the vertical and horizontal subspaces. Then
\[
  \boxed{
  P^V = G P_V G^{-1} = \begin{bmatrix} 0 & 0 \\ 0 & 1 \end{bmatrix} \qquad \text{and} \qquad P^H = G P_H G^{-1} = \begin{bmatrix} I & 0 \\ 0 & 0 \end{bmatrix},
  }
\]
where in both expressions the second equality holds only on $N \Sigma_\delta$ because of the explicit form of the metric.

The mapping $d\pi_{\sigma,n} : T_{\sigma,n} N \Sigma \to T_\sigma \Sigma$ maps $\partial/\partial x_j \in T_{\sigma,n} N \Sigma$ to $\partial/\partial x_j \in T_\sigma \Sigma$ and maps $\partial/\partial y \in T_{\sigma,n} N \Sigma$ to $0 \in T_\sigma \Sigma$. Thus $J_{\sigma,n}$ (which is defined in \eqref{J}) maps $dx_j \in T_{\sigma,n}^* N \Sigma$ to $dx_j \in T_\sigma^* \Sigma$.

Therefore, if $(\sigma, n, \xi, \eta) \in T^* N \Sigma_\delta$ has coordinates $(x, y, p, r) \in (\R^2 \times \R)^2$, then $\xi$ and $\eta$ in $T_{\sigma,n}^* N \Sigma_\delta$ have coordinates in $\R^2 \times \R$ given by
\[
  P^H \begin{bmatrix} p \\ r \end{bmatrix} = \begin{bmatrix} p \\ 0 \end{bmatrix} \qquad \text{and} \qquad P^V \begin{bmatrix} p \\ r \end{bmatrix} = \begin{bmatrix} 0 \\ r \end{bmatrix}.
\]
Consequently, $J_{\sigma,n} \xi \in T_\sigma^* \Sigma$ has coordinates $p \in \R^2$.

\subsection{Hamiltonians} \label{s:H}

In this subsection, we first introduce some notation, then we show how to implement the gauge transformation mentioned earlier, and finally we write the Hamiltonians $H_\lambda$ and $L_\lambda$ in local coordinates.

Define $g_\Sigma(x) = \det(G_\Sigma(x))$ and $g(x,y) = \det(G(x,y))$. Set
\[
  D_x = \begin{bmatrix} i \partial/\partial x_1 \\ i \partial/\partial x_2 \end{bmatrix}, \qquad D_y = i \partial/\partial y, \qquad D = \begin{bmatrix} D_x \\ D_y \end{bmatrix}.
\]
Denote by $X^T$ the transpose of the matrix or column vector $X$. Denote by $A = [A_1 \ A_2 \ A_3]^T$ the local coordinates of the magnetic vector potential as a \emph{cotangent} vector field on $N \Sigma$. Let $A_H$ and $A_V$ be defined by the equations
\begin{equation} \label{AHV}
  \begin{bmatrix} A_H \\ Y \end{bmatrix} = P^H A \qquad \text{and} \qquad \begin{bmatrix} X \\ A_V \end{bmatrix} = P^V A.
\end{equation}
(Note that $Y$ and $X$ are equal to zero on $N\Sigma_\delta$.)
Let $A' = A - d\gamma$, and let $A_H'$ and $A_V'$ be defined similarly as in \eqref{AHV} with $A$ replaced by $A'$. Note that, on $N \Sigma_\delta$, we have $A_H = [A_1 \ A_2]^T$ and $A_V = A_3$, so that $A = [A_H \ A_V]^T$.

We next describe how to implement the gauge transformation mentioned earlier. More precisely, we will prove the following lemma:

\begin{lem}[Gauge condition]
  Let $A$ be a smooth cotangent vector field on $N\Sigma$, and let $P^H$ and $P^V$ as above. Then there exists a real valued function $\gamma \in C^\infty(N\Sigma)$ such that $P^V(A - \partial \gamma)=0$ and $P^H \partial \gamma = e_1$ on $N \Sigma_\delta$, where $e_1$ is smooth and vanishes to first order in $y$ at $(x,0)$ (in local coordinates).
\end{lem}

\begin{proof}
Let $\{\chi_j\}$ be the partition of unity given in \eqref{pu}. We set $\gamma = \sum_{j=1}^m \chi_j^2 \gamma_j^{}$ with $\gamma_j$ defined on each chart $(\mathcal{V}_j, (x,y))$ on $N \Sigma$ by $\gamma_j(x,y) = \int_0^y A_V(x,s) ds$. Thus $\gamma$ is well-defined on $N\Sigma$. Let us show that
\[
  \gamma(x,y) = \int_0^y A_3(x,s) ds
\]
on every chart on $N \Sigma_\delta$. Let $(\mathcal{V}_j, (x,y))$ and $(\mathcal{V}_k, (\tilde{x}, \tilde{y}))$ be two different charts on $N \Sigma$ that overlap each other. Recall that $A_V = A_3$ on $N\Sigma_\delta$. Note that, on $\mathcal{V}_j \cap \mathcal{V}_k \cap N \Sigma_\delta$, by performing a coordinate transformation (for covectors), we have
\[
  \begin{split}
    \gamma_j(x,y) & = \int_0^y A_3(x,s) ds = \int_0^{\tilde{y}} \Big( \sum_{l=1}^2 \frac{\partial \tilde{x}_l}{\partial y} \tilde{A}_l(\tilde{x}, s) + \frac{\partial \tilde{y}}{\partial y} \tilde{A}_3(\tilde{x}, s) \Big) ds \\
    & = \int_0^{\tilde{y}} \tilde{A}_3(\tilde{x}, s) ds = \gamma_k(\tilde{x}, \tilde{y}),
  \end{split}
\]
where we used that $\partial \tilde{x}_l/\partial y = 0$ for $l=1,2$, and $\partial \tilde{y}/\partial y = 1$. Thus $\gamma_j$ is invariant by coordinate transformations on $N\Sigma_\delta$. This leads to the above expression for $\gamma$ and thus proves the claim.

We finally get to the gauge transformation. Recall that $A = [A_H \ A_V]^T$ on $N\Sigma_\delta$. Then with the above definition for $\gamma$, the substitution $A \leadsto A - \partial \gamma$ translates into
\[
  \boxed{
  \begin{bmatrix} A_H \\ A_V \end{bmatrix} \leadsto \begin{bmatrix} A_H' \\ 0 \end{bmatrix},
  }
\]
where
\[
  A_H'(x,y) \coloneqq A_H^{}(x,y) - \int_0^y \partial_x^{} A_V^{}(x,s) ds = A_H(x,y) + e_1.
\]
Here, the last equality follows by Taylor's Theorem.
\end{proof}

We next give the local coordinate expressions for the Hamiltonians.

The local coordinate expression in $N \Sigma$ for the operator $H_\lambda$ given in \eqref{Hlam}~is
\[
  H_\lambda = (g^{-1/2} D g^{1/2} + A)^T G^{-1} (D + A) + V + \lambda^4 W.
\]
Thus in $N \Sigma_\delta$ we may use the explicit form of the metric to obtain
\[
  \begin{split}
    H_\lambda & = (g^{-1/2} D_x g^{1/2} + A_H)^T(G_\Sigma + C)^{-1}(D_x + A_H) + V \\
    & \quad + (g^{-1/2} D_y g^{1/2} + A_V)^T(D_y + A_V) + \lambda^4 W.
  \end{split}
\]
Consequently
\[
  \begin{split}
    S_\gamma^* H_\lambda S_\gamma & = (g^{-1/2} D_x g^{1/2} + A_H')^T(G_\Sigma + C)^{-1}(D_x + A_H') + V \\
    & \quad + g^{-1/2} D_y^T g^{1/2} D_y + \lambda^4 W.
  \end{split}
\]

We now calculate the local coordinate expression for $L_\lambda^{} = U_\lambda^* S_\gamma^* H_\lambda^{} S_\gamma U_\lambda^{}$ with $U_\lambda = D_\lambda M_\lambda$ (recall the definitions in Section \ref{s:main}). The densities on $N \Sigma$ are given in local coordinates by
\begin{equation} \label{dens}
  \begin{aligned}
    \dvol & = \sqrt{g(x,y)} |dx \wedge dy|, \\
    \dvol_\lambda & = \sqrt{g(x,y/\lambda)} |dx \wedge dy|, \\
    \dvol_{N\Sigma} & = \sqrt{g(x,0)} |dx \wedge dy|.
  \end{aligned}
\end{equation}
Thus the transformation $M_\lambda$ is multiplication by
\[
  m_\lambda(x,y) \coloneqq \left[ \frac{g(x,0)}{g(x,y/\lambda)} \right]^{1/4}.
\]
When conjugation by $D_\lambda$ acts on functions, it results into composition with $d_\lambda^{-1}$. When it acts on $D_y$, it replaces it by $\lambda D_y$. Conjugation by $M_\lambda$ results in multiplication by $m_\lambda$ to the right, and multiplication by $m_\lambda^{-1}$ to the left. Set
\begin{gather*}
  G_\lambda(x,y) = \begin{bmatrix} I & 0 \\ 0 & \lambda^{-1} \end{bmatrix} G(x,y/\lambda) \begin{bmatrix} I & 0 \\ 0 & \lambda^{-1} \end{bmatrix}, \\
  A_\lambda^{(\cdot)}(x,y) = \begin{bmatrix} I & 0 \\ 0 & \lambda^{-1} \end{bmatrix} A^{(\cdot)}(x,y/\lambda),
\end{gather*}
\vspace{-2em}
\begin{alignat*}{2}
  C_\lambda(x,y) & = C(x,y/\lambda), & \qquad V_\lambda(x,y) & = V(x,y/\lambda), \\
  W_\lambda(x,y) & = W(x,y/\lambda), & \qquad k_\lambda(x,y) & = \log m_\lambda(x,y),\\
  g_\lambda(x,y) & = g(x,y/\lambda), & \qquad g_\infty(x,y) & = g(x,0).
\end{alignat*}
Let $A_{\lambda H}$ and $A_{\lambda V}$ be defined as in \eqref{AHV} with $A$ replaced by $A_\lambda$. For any function $f$ on $N \Sigma$ define, in local coordinates,
\begin{gather*}
  \partial f \coloneqq [ \partial_x f \ \partial_y f]^T \coloneqq [ \partial f/\partial x_1 \ \partial f/\partial x_2 \ \partial f/\partial y ]^T, \\
  \partial^2_y f \coloneqq \partial^2 f/\partial y^2.
\end{gather*}
Let $X$ be a real valued multiplication operator. We will denote by $D_{(\cdot)}^*$ and $X^*$ the formal adjoints of $D_{(\cdot)}$ and $X$ with respect to the inner product $\langle \psi, \phi \rangle = \int_{N \Sigma} \overline{\psi} \varphi \dvol_{N \Sigma}$ on $L^2(N \Sigma, \dvol_{N \Sigma})$. Thus $D_x^* = g_\infty^{-1/2} D_x^T g_\infty^{1/2}$, $D_y^* = g_\infty^{-1/2} D_y^T g_\infty^{1/2} = D_y^T = D_y$, and $X^* = g_\infty^{-1/2} X^T g_\infty^{1/2} = X^T$. (Note that $D_{(\cdot)}$ contains a factor $i$.) Hence
\begin{equation} \label{Llam1}
  \begin{split}
    & L_\lambda^{} = m_\lambda^{-1} (g_\lambda^{-1/2} D g_\lambda^{1/2} + A_\lambda')^T G_\lambda^{-1} (D + A_\lambda') m_\lambda^{} + V_\lambda^{} + \lambda^4 W_\lambda^{} \\
    & = (D + i \partial k_\lambda^{} + A_\lambda')^* G_\lambda^{-1} (D + i\partial k_\lambda^{} + A_\lambda') + V_\lambda^{} + \lambda^4 W_\lambda^{} \\
    & = (D + A_\lambda')^* G_\lambda^{-1} (D +\! A_\lambda') + V_\lambda^{} + \lambda^4 W_\lambda^{} + \partial k_\lambda^* G_\lambda^{-1} \partial k_\lambda^{} + \{D^* G_\lambda^{-1} i\partial k_\lambda^{}\},
  \end{split}
\end{equation}
where in the last equality we used that $\partial k_\lambda^* G_\lambda^{-1} A_\lambda' = A_\lambda'^* G_\lambda^{-1} \partial k_\lambda^{}$. Here, terms of the form $\{ D_{(\cdot)}^* \cdots\}$ (written with curly brackets to distinguish from others) are multiplication operators (we will use this notation from now on). Thus in $N \Sigma_{\lambda \delta}$ we may use the explicit form of the metric to obtain
\begin{equation} \label{Llam2}
  \begin{split}
    L_\lambda^{} & = (g_\Sigma^{-1/2} m_\lambda^{} D_x^{} m_\lambda^{-1} g_\Sigma^{1/2} + A_{\lambda H}')^T (G_\Sigma^{} + C_\lambda^{})^{-1} (m_\lambda^{-1} D_x^{} m_\lambda^{} + A_{\lambda H}') + V_\lambda \\
    & \quad + \lambda^2 \big[ m_\lambda^{} D_y^T m_\lambda^{-1} m_\lambda^{-1} D_y^{} m_\lambda^{} + \lambda^2 W_\lambda^{} \big] \\
    & = (D_x^{} + A_{\lambda H}')^* (G_\Sigma^{} + C_\lambda^{})^{-1} (D_x^{} + A_{\lambda H}') + V_\lambda^{} \\
    & \quad + \partial_x^{} k_\lambda^* (G_\Sigma^{} + C_\lambda^{})^{-1} \partial_x^{} k_\lambda^{} + \{D_x^* (G_\Sigma^{} + C_\lambda^{})^{-1} i\partial_x^{} k_\lambda^{}\} \\
    & \quad + \lambda^2 \big[ D_y^* D_y^{} + (\partial_y^{} k_\lambda^{})^2 - \partial^2_y k_\lambda^{} + \lambda^2 W_\lambda^{} \big].
  \end{split}
\end{equation}

We can now comment on the origin of the Hamiltonian $H_\Sigma + \lambda^2 H_{O,\lambda}$. By setting $C_\lambda = 0$ and $m_\lambda = 1$ in the above expression for $L_\lambda$, we obtain the local coordinate expression for the following operator acting on a function:
\[
  i \divv_\lambda( i \grad_\lambda(\psi) + A_\lambda' \psi ) + \langle A_\lambda', i \grad_\lambda(\psi) + A_\lambda' \psi \rangle_\lambda + V_\lambda + \lambda^4 W_\lambda.
\]
Note the presence of the metric defined by $\langle \cdot, \cdot \rangle_\lambda$. Further expansion of $m_\lambda$ and the potentials will reveal that, to leading order, this operator is exactly $H_\Sigma + \lambda^2 H_{\lambda, O}$.

\subsection{\texorpdfstring{Large $\lambda$ expansions}{Large lambda expansion}}

Our next step is performing a large $\lambda$ expansion in $L_\lambda$.

To keep track of error terms we will write $E_j$ to denote a smooth (possibly matrix valued) function of $(x,y)$ that vanishes to $j$th order in $y$ at $(x,0)$, evaluated at $(x,y/\lambda)$. (The function $E_j$ behaves like $(y/\lambda)^j$ for small $y/\lambda$.) The derivatives of these functions behave as
\begin{equation} \label{Ej}
  \partial_x E_j = E_j, \qquad \partial_y E_0 = \lambda^{-1} E_0, \qquad \partial_y E_j = \lambda^{-1} E_{j-1} \ \text{for} \ j \ge 1.
\end{equation}

The Gaussian curvature and the mean curvature on $\Sigma$ are defined by
\[
  s = \det(L) \qquad \text{and} \qquad h = \frac{1}{2} \tr(L).
\]
We set
\begin{gather*}
  w(\sigma) = \nu(\sigma)^T (\hess_n W)(\sigma,0) \nu(\sigma), \\
  f_1(\sigma) = \sum_{|\alpha|=3} a_\alpha(\sigma) \nu(\sigma)^\alpha \qquad \text{and} \qquad f_2(\sigma) = \sum_{|\beta|=4} b_\beta(\sigma) \nu(\sigma)^\beta, \\
  F_\lambda = \lambda^{-1} f_1 y^3 + \lambda^{-2} f_2 y^4.
\end{gather*}

We will need the following expansions for our proofs.

\begin{lem}[Expansion formulae] \label{l:exp}
  Assume the hypotheses of Theorem \ref{t:main}, let $(\mathcal{V}, (x,y))$ be a chart on $N \Sigma_{\lambda \delta}$, and let $\|\cdot\|$ be the operator norm (induced by the Euclidean norm). Then for all $(x,y/\lambda)$ with $|y/\lambda| \|G_\Sigma^{-1}[\langle \sigma_i, L \sigma_j \rangle]\| < 1$, we have
  \[
    \begin{gathered}
      A_{\lambda H}' = A_H^{}(x,0) + E_1^{}, \\
      V_\lambda = V(x,0) + E_1, \\
      (G_\Sigma^{} + C_\lambda^{})^{-1} = G_\Sigma^{-1} + E_1^{},
    \end{gathered}
  \]
  and
  \begin{equation} \label{l:exp2}
    \begin{gathered}
      \lambda^2 W_\lambda = \tfrac{1}{2} w \, y^2 + F_\lambda + \lambda^2 E_5, \\
      \partial_y \lambda^2 W_\lambda = w y + \lambda E_2, \\
      \partial_{x_j} W_\lambda = E_4, \\
      y < C(1 + \lambda^2 W_\lambda)^{1/2},
    \end{gathered}
  \end{equation}
  and
  \begin{equation} \label{l:exp3}
    \begin{gathered}
      \partial_x k_\lambda = E_1, \\
      \partial_y^j k_\lambda^{} = \lambda^{-j} E_0^{} \ \text{for} \ 1 \le j \le 2, \\
      \lambda^2 \Big[ (\partial_y^{} k_\lambda^{})^2 - \partial^2_y k_\lambda^{} \Big] = K + E_1^{}, \\
      \partial_x k_\lambda^* (G_\Sigma + C_\lambda)^{-1} \partial_x k_\lambda^{} + \{ D_x^* (G_\Sigma + C_\lambda)^{-1} i \partial_x k_\lambda^{} \} = E_1^{}, \\
      \partial k_\lambda^* G_\lambda^{-1} \partial k_\lambda^{} + \{ D^* G_\lambda^{-1} i \partial k_\lambda^{} \} = K + E_1^{}.
    \end{gathered}
  \end{equation}
\end{lem}

\begin{proof}
Calculating the Taylor series expansions around $(x,0)$ we obtain the formulae for $A_{\lambda H}'$, $V_\lambda$, and $\lambda^2 W_\lambda$, and the estimates for $y$ and the derivatives of $W_\lambda$. To prove the other formulae we use the expressions in Lemma \ref{l:metric} evaluated at $(x,y/\lambda)$. The formula for $(G_\Sigma + C_\lambda)^{-1}$ is immediate. The remaining formulae follow by a short calculation using
\[
  k_\lambda(x,y) = \frac{y}{2 \lambda} \tr(L) + \frac{y^2}{4 \lambda^2} \tr(L^2) + E_3,
\]
the error bounds in \eqref{Ej}, and the identity $\tr(L)^2 - \tr(L^2) = 2\det(L)$. We omit the details.
\end{proof}

Thus, substituting the expansions in Lemma \ref{l:exp} into \eqref{Llam2}, we obtain, in $N \Sigma_{\lambda \delta}$,
\begin{equation} \label{Llame}
  \boxed{
  L_\lambda = L_{0,\lambda} + Q_\lambda,
  }
\end{equation}
where
\[
  L_{0,\lambda} = H_\Sigma + \lambda^2 H_{O,\lambda}
\]
with
\begin{equation} \label{Hs}
  \begin{gathered}
    H_\Sigma^{} = (D_x^{} + A^{}_H(x,0))^* G_\Sigma^{-1} (D_x^{} + A^{}_H(x,0)) + V(x,0) + K, \\
    H_{O,\lambda}^{} = D_y^* D_y^{} + \lambda^2 W_\lambda^{} = D_y^* D_y^{} + \tfrac{1}{2} w \, y^2 + F_\lambda + \lambda^2 E_5^{},
  \end{gathered}
\end{equation}
and
\[
  Q_\lambda^{} = D_x^* E_1^{} D_x + D_x^* E_1^{} + E_1^* D_x^{} + E_1^{}.
\]

To summarize so far, we have written the quantities that we need for our proofs in local coordinates, and we have proved some expansion formulae which yield the large $\lambda$ expansion for $L_\lambda$. We have also proved that the gauge condition can be fulfilled.

\section{\texorpdfstring{Transfer from $\R^3$ to $N\Sigma$}{Transfer from R3 to NSigma}} \label{s:transfer}

For $x \in \R^3$, define the distance from $x$ to $\Sigma$ by
\[
  d(x,\Sigma) = \inf \set{|x-\sigma|}{\sigma \in \Sigma},
\]
and for $c > 0$ set
\[
  U_c = \set{x \in \R^3}{d(x,\Sigma) < c}.
\]
Then $U_\delta$ is the tubular neighborhood of $\Sigma$ that is diffeomorphic to $N \Sigma_\delta$. Let $\M$ be either $\R^3$ or $N\Sigma$, and let $\M_\delta$ be either $U_\delta$ or $N\Sigma_\delta$. Let $\dvol$ be either the density defined by the extended metric on $N\Sigma$ or the Lebesgue measure on $\R^3$. We will denote by $H_\lambda$ the operator given by the expression in \eqref{Hlam} acting on $L^2(\M,\dvol)$, and by $H_\lambda^\delta$ the operator given by the expression in \eqref{Hlam} acting on $L^2(\M_\delta,\dvol)$ with Dirichlet boundary conditions on $\partial M_\delta$.

\begin{prop} \label{p:transfer}
  Assume the hypothesis of Theorem \ref{t:main}. Let $\psi \in L^2(\M,\dvol)$ with $\| H_\lambda \psi \| \le C_1 \lambda^2$ and $\| \psi \|=1$. Let $H_\lambda$ and $H_\lambda^\delta$ be the operators defined above. Then given $\varepsilon > 0$, there exists a constant $C_2 > 0$, depending only on $C_1$ and $\varepsilon$, such that, for all $t \in \R$,
  \begin{equation} \label{t1}
    \| F_{\M \setminus \M_\varepsilon} e^{-itH_\lambda} \psi \| \le C_2 \lambda^{-1}.
  \end{equation}
  Here $F_X$ is the operator of multiplication by the characteristic function with support in $X$. Furthermore, given $T>0$, there exists a constant $C_3>0$ such that, for all $0 < \varepsilon < \delta$,
  \begin{equation} \label{t2}
    \sup_{t \in [0,T]} \| F_{\M_\varepsilon} e^{-itH_\lambda} \psi - e^{-itH_\lambda^\delta} F_{\M_\varepsilon} \psi \| \le C_3 \lambda^{-1/4}.
  \end{equation}
  The constant $C_3$ depends only on $\varepsilon$, $C_1$ and $T$.
\end{prop}

\begin{proof}
Let us show \eqref{t1}. By the Cauchy-Schwarz inequality, we have
\[
  \langle \psi, H_\lambda \psi \rangle \le C_1 \lambda^2.
\]
Without loss of generality, we can assume that $V \ge 0$ (see the remark in the paragraph above \eqref{eta1}). Consequently
\begin{equation} \label{upperb}
  \begin{gathered}
    \| \grad(\psi) + A \psi \|^2 \le C_1 \lambda^2, \\
    \langle \psi, W \psi \rangle \le C_1 \lambda^{-2}.
  \end{gathered}
\end{equation}
Thus
\[
  \langle \psi, F_{\M \setminus \M_\varepsilon} \psi \rangle \le (\varepsilon^2 \kappa)^{-1} \langle \psi, F_{\M \setminus \M_\varepsilon} W \psi \rangle \le (\varepsilon^2 \kappa)^{-1} C_1 \lambda^{-2}.
\]
This proves \eqref{t1} because $e^{-itH_\lambda} \psi$ obeys the same hypotheses as $\psi$.

To prove \eqref{t2}, let $\eta$ be a smooth function on $\M$ obeying $0 \le \eta \le F_{\M_{\varepsilon/2}}$ and $\eta = 1$ on $\M_{\varepsilon/4}$. It suffices to show that
\[
  \sup_{t \in [0,T]} \| e^{itH_\lambda^\delta} \eta e^{-itH_\lambda} \psi - \eta \psi \| \le C_3 \lambda^{-1/4}.
\]
This is true because the terms that arise from replacing $F_{\M_\varepsilon}$ by $\eta$ in \eqref{t2} decay like $1/\lambda$ (this follows by \eqref{t1} with $\varepsilon/2$). Set
\[
  \varphi_{t,\lambda} = e^{itH_\lambda^\delta} \eta e^{-itH_\lambda} \psi - \eta \psi.
\]
Then, by writing $\varphi_{t,\lambda}$ as the integral of its derivative, we obtain
\[
  \varphi_{t,\lambda} = i \int_0^t e^{is H_\lambda^\delta} (H_\lambda^\delta \eta - \eta H_\lambda) e^{-isH_\lambda} \psi \, ds,
\]
where
\[
  (H_\lambda^\delta \eta - \eta H_\lambda) \psi = 2 \langle i \grad(\eta), i \grad(\psi) + A \psi \rangle - \divv(\grad(\eta)) \psi.
\]
Let $\{\chi_j\}$ be a partition of unity as in \eqref{pu}, and let $\tilde{\eta}$ be a smooth function on $\M$ that equals $1$ on the support of $\grad(\eta)$ and $0$ on $\M_{\varepsilon/8}$. Then
\begin{equation} \label{ta}
  \begin{split}
    & \| \varphi_{t,\lambda} \|^2 = i \int_0^t \sum_{j=1}^m \langle \tilde{\eta} e^{-is H_\lambda^\delta} \varphi_{t,\lambda}, \tilde{\eta} \chi_j^2 (H_\lambda^\delta \eta - \eta H_\lambda) e^{-isH_\lambda} \psi \rangle ds \\
    & \le C \int_0^t \sum_{j=1}^m \| \chi_j \tilde{\eta} e^{-is H_\lambda^\delta} \varphi_{t,\lambda} \| ( \| \chi_j \tilde{\eta} \{D^T \eta\} G^{-1} (D + A) e^{-isH_\lambda} \psi \| + 1 )ds.
  \end{split}
\end{equation}
Thus, we are left to estimating the two quantities in \eqref{ta}.

Let us consider the second quantity in \eqref{ta}. Let $\mathcal{A} = \M_{\varepsilon/2} \setminus \M_{\varepsilon/4}$, and let $\chi \in \{\chi_j\}$. We will show that
\begin{equation} \label{tb}
  \| F_\mathcal{A} \chi G^{-1/2} (D + A) \psi \| \le C \lambda^{1/2}.
\end{equation}
To establish this, we begin observing that
\[
  \| F_\mathcal{A} \chi G^{-1/2}(D+A) \psi \| = \| F_\mathcal{A} \chi G^{-1/2}(D+A) \tilde{\eta} \psi \| \le \| \chi G^{-1/2}(D+A) \tilde{\eta} \psi \|.
\]
Now, using the Cauchy-Schwarz inequality, and then \eqref{t1}, we find
\[
  \begin{split}
    \| \chi G^{-1/2} (D+A) \tilde{\eta} \psi \|^2 & \le 2 \| \tilde{\chi} \tilde{\eta} \psi \| \| \chi (D+A)^* G^{-1} (D+A) \tilde{\eta} \psi \| + C \\
    & \le C \lambda^{-1} \| \chi (D+A)^* G^{-1} (D+A) \tilde{\eta} \psi \| + C,
  \end{split}
\]
where $\tilde{\chi}$ is like $\chi$, but it is supported on a slightly larger subset. Thus \eqref{tb} follows from
\begin{equation} \label{tc}
  \| \chi (D+A)^* G^{-1} (D+A) \tilde{\eta} \psi \| \le C \lambda^2.
\end{equation}
To prove \eqref{tc}, write
\begin{gather*}
  T = (g^{-1/2} D g^{1/2}+A)^T G^{-1} (D+A), \\
  F = V + \lambda^4 W \qquad \text{and} \qquad S = \{g^{-1/2} D^T g^{1/2} G^{-1} D \chi^2 \tilde{\eta}^2 F\}.
\end{gather*}
Then calculating and exploiting positivity we obtain
\[
  \begin{split}
    H_\lambda^2 & \ge H_\lambda \chi^2 H_\lambda \\
    & = T \chi^2 T + 2(g^{-1/2} D g^{1/2} + A)^T G^{-1/2} \chi F \chi G^{-1/2}(D+A) + (\chi F)^2 + S \\
    & \ge T \chi^2 T - C\lambda^4.
  \end{split}
\]
Hence
\[
  \begin{split}
    \| \chi T \tilde{\eta} \psi \|^2 & \le C \lambda^4 \| \tilde{\eta} \psi \|^2 + \| H_\lambda \tilde{\eta} \psi \|^2 \\
    & \le C \lambda^4 + C \| H_\lambda \psi \|^2 + C \| \grad(\psi) + A \psi \|^2 + C \le C \lambda^4,
  \end{split}
\]
where we used \eqref{upperb} in the last inequality. This proves \eqref{tc} (and thus \eqref{tb}).

We next consider the first quantity in \eqref{ta}. Using the Cauchy-Schwarz inequality and calculating we find
\[
  \begin{split}
    \langle \varphi_{t,\lambda}, H_\lambda^\delta \varphi_{t,\lambda} \rangle & \le 3 \langle \eta e^{-itH_\lambda} \psi, H_\lambda^\delta \eta e^{-itH_\lambda} \psi \rangle + 3 \langle \eta \psi, H_\lambda^\delta \eta \psi \rangle \\
    & = \tfrac{3}{2} \langle e^{-itH_\lambda} \psi, ( H_\lambda \eta^2 + \eta^2 H_\lambda + 2 |\grad(\eta)|^2 ) e^{-itH_\lambda} \psi \rangle \\
    & \quad + \tfrac{3}{2} \langle \psi, ( H_\lambda \eta^2 + \eta^2 H_\lambda + 2 |\grad(\eta)|^2 ) \psi \rangle \\
    & \le C \lambda^2.
  \end{split}
\]
Now, by proceeding exactly as in the proof of \eqref{t1}, we conclude that
\begin{equation} \label{td}
  \| \tilde{\eta} e^{-isH_\lambda^\delta} \varphi_{t,\lambda} \| \le C \lambda^{-1}.
\end{equation}

Therefore, substituting \eqref{td} and \eqref{tc} into \eqref{ta}, we obtain
\[
  \| \varphi_{t,\lambda} \|^2 \le CmT \lambda^{-1} \lambda^{1/2} \le C \lambda^{-1/2}.
\]
This proves \eqref{t2} and completes the proof.
\end{proof}

\section{Energy bounds} \label{s:ebounds}

In this section we prove some estimates which form a basic ingredient in the proof of Theorem \ref{t:main}.

Let $L_{\sharp \lambda}$ be either $L_\lambda$ or $L_{0,\lambda} = H_\Sigma + \lambda^2 H_{O,\lambda}$. We define
\[
  R_{\sharp \lambda} = (1 + \lambda^{-2} L_{\sharp \lambda})^{-1}.
\]
Here both operators act on $L^2(N \Sigma, \dvol_{N \Sigma})$. Observe that, without loss of generality, we may assume that $V > C_1$ for a constant $C_1>0$ such that $K > -C_1$ so that $L_{\sharp \lambda}$ is always positive and thus $R_{\sharp \lambda}$ is well-defined. (Indeed, this condition can be fulfilled by introducing the unitary transformation $e^{iC_2t}$, for a constant $C_2>0$, which leads to the substitution $V_\lambda \leadsto V_\lambda + C_2$.)

We will use a cutoff in the $n$ variable. Let
\begin{equation} \label{eta1}
  \eta_1^{} = \eta_1^{}(n)
\end{equation}
be a smooth function on $N \Sigma$ supported in $N \Sigma_{\lambda \delta}$ that equals $1$ on $N \Sigma_{\lambda \delta_1}$ with $0 < \delta_1 < \delta$ and $0 \le \eta_1 \le 1$.

The following lemma will be used many times in our analysis. It shows that, in each chart on $N \Sigma_{\lambda \delta}$, the operators $\lambda^{-1} D_x$, $D_y$, $(1+\lambda^2 W_\lambda)^{1/2}$, and some of its powers are bounded (in the operator sense) by powers of $1+\lambda^{-2} L_{\sharp \lambda}$.

\begin{lem}[Energy bounds] \label{l:e}
  Assume the hypotheses of Theorem \ref{t:main}. Let $(\mathcal{V}, (x,y))$ be a chart on $N \Sigma_{\lambda \delta}$, and let $\chi = \chi(\sigma)$ be a smooth function on $N\Sigma$ supported in $\mathcal{V}$ with $0 \le \chi \le 1$. Then there are constants $\lambda_0$ and $C$ such that, for every $\lambda > \lambda_0$, the following holds:
  \begin{enumerate}[\rm (i)]
    \item \label{l:e1}
      $\| \chi \eta_1^{} \lambda^{-1} D_x^{} R_{\sharp \lambda}^{1/2} \| + \| \chi \eta_1^{} D_y^{} R_{\sharp \lambda}^{1/2} \| + \| (1+\lambda^2 W_\lambda^{})^{1/2} R_{\sharp \lambda}^{1/2} \| \le C$.
    \item \label{l:e2}
      If $l$ and $p$ are integers and $\alpha$ is a multi-index such that $l \ge 0$, $p \ge 0$ and $l + |\alpha| + p \le 2$, then $\| \chi \eta_1^{} (1 + \lambda^2 W_\lambda^{})^{l/2} (\lambda^{-1} D_x^{})^\alpha D_y^p R_{\sharp \lambda} \| \le C$.
    \item \label{l:e3}
      If $l$ and $p$ are integers and $\alpha$ is a multi-index such that $l > 0$, $p \ge 0$ and $|\alpha| + p \le 2$, then $\| \chi \eta_1^{} (1 + \lambda^2 W_\lambda^{})^{l/2} (\lambda^{-1} D_x^{})^\alpha D_y^p R_{\sharp \lambda}^{l+1} \| \le C$.
  \end{enumerate}
  Here $(\lambda^{-1} D_x)^\alpha D_y^p = (\lambda^{-1} i \partial/ \partial x_1)^{\alpha_1} (\lambda^{-1} i \partial/\partial x_2)^{\alpha_2} (i\partial/\partial_y)^p$.
\end{lem}

\begin{rem} \label{rem}
  We will denote by $C$ a positive constant that may change from line to line. We will prove the lemma for $L_\lambda$; the proof for $L_{0,\lambda}$ is similar. In fact, we can see from \eqref{Llame} that $L_\lambda$ contains the term $L_{0,\lambda}$. For the sake of clarity, we will prove the Lemma for arbitrary $\gamma$ in $L_\lambda$; that is, we will prove the Lemma without assuming that $A_V=0$. This is just to make clear that the statement is true independently of the gauge transformation. To simplify the notation we write $A$ in place of $A'$.
\end{rem}

\begin{proof}
See Remark \ref{rem}. Let $f = \chi \eta_1$. Then $f \in C_0^\infty(N \Sigma)$ and $0 \le f \le 1$. By writing $G_\Sigma^{-1}$ as $G_\Sigma^{-1/2} G_\Sigma^{-1/2}$, it follows easily that $G_\lambda^{-1} \ge f G_\lambda^{-1} f$. Thus, by recalling the expression for $L_\lambda$ in \eqref{Llam1},
\[
  L_\lambda^{} \ge (D + i\partial k_\lambda^{} + A_\lambda^{})^* f G_\lambda^{-1} f (D + i\partial k_\lambda^{} + A_\lambda^{}) + V_\lambda^{} + \lambda^4 W_\lambda^{}.
\]
Using the explicit form of the metric we obtain
\[
  f G_\lambda^{-1} f = f \begin{bmatrix} (G_\Sigma + C_\lambda)^{-1} & 0 \\ 0 & \lambda^2 \end{bmatrix} f \ge C^{-1} f \begin{bmatrix} I & 0 \\ 0 & \lambda^2 \end{bmatrix} f.
\]
Using \eqref{l:exp3}, and observing that $\eta_1(y) = \phi(y/\lambda)$ for $\phi \in C_0^\infty(\R)$, by a short calculation we prove that
\[
  \partial k_\lambda^* f G_\lambda^{-1} f \partial k_\lambda + \{ D^* f G_\lambda^{-1} f i \partial k_\lambda \} \le K + E_1 + 2 \chi^2 \lambda^2 \eta_1 \{D_y \eta_1\} i \partial_y k_\lambda \le C.
\]
Hence, for sufficiently large $\lambda$,
\[
  \begin{split}
    \lambda^{-2} L_\lambda^{} & \ge \lambda^{-2} (D + A_\lambda^{})^* f G_\lambda^{-1} f (D + A_\lambda^{}) + \lambda^2 W_\lambda^{} - C \lambda^{-2} \\
    & \ge C^{-1} (D_x + A_{\lambda H})^* \lambda^{-2} f^2 (D_x + A_{\lambda H}) \\
    & \quad + C^{-1} (D_y + A_{\lambda V})^* f^2 (D_y + A_{\lambda V}) + \lambda^2 W_\lambda^{} - 1.
  \end{split}
\]
Thus, using the Cauchy-Schwarz inequality,
\[
  \begin{split}
    & \langle \psi, C(1 + \lambda^{-2} L_\lambda) \psi \rangle \\
    & \ge \| f \lambda^{-1} (D_x + A_{\lambda H}) \psi \|^2 + \| f(D_y + A_{\lambda V}) \psi \|^2 + \| \lambda W_\lambda^{1/2} \psi \|^2 \\
    & \ge 2^{-1} \| f \lambda^{-1} D_x \psi \|^2 + 2^{-1} \| f D_y \psi \|^2 + \| \lambda W_\lambda^{1/2} \psi \|^2 \\
    & \quad -7 (\lambda^{-2} \| f A_{\lambda H} \|_{L^\infty(N\Sigma)}^2 + \| f A_{\lambda V} \|_{L^\infty(N\Sigma)}^2 ) \| \psi \|^2.
  \end{split}
\]
Recall that $A_{\lambda V}$ contains a factor $\lambda^{-1}$. Then it follows from this inequality that, for sufficiently large $\lambda$,
\[
  D_x^* \lambda^{-2} f^2 D_x^{} + D_y^* f^2 D_y^{} + 1 + \lambda^2 W_\lambda^{} \le C (1 + \lambda^{-2} L_\lambda^{}).
\]
Since $R_\lambda^{1/2} (1 + \lambda^{-2} L_\lambda^{}) R_\lambda^{1/2} = 1$, we get
\[
  \| f \lambda^{-1} D_x^{} R_\lambda^{1/2} \| + \| f D_y^{} R_\lambda^{1/2} \| + \| (1+\lambda^2 W_\lambda^{})^{1/2} R_\lambda^{1/2} \| \le C. 
\]
This proves part \eqref{l:e1}.

To prove part \eqref{l:e2}, write \eqref{Llam1} as
\[
  L_\lambda = D^* G_\lambda^{-1} D + T + U
\]
with
\begin{gather*}
  T \coloneqq A_\lambda^* G_\lambda^{-1} D + D^* G_\lambda^{-1} A_\lambda^{}, \\
  U \coloneqq A_\lambda^* G_\lambda^{-1} A_\lambda^{} + V_\lambda^{} + \lambda^4 W_\lambda^{} + \partial k_\lambda^* G_\lambda^{-1} \partial k_\lambda^{} + \{D^* G_\lambda^{-1} i \partial k_\lambda^{}\}.
\end{gather*}
Then calculating we find
\begin{equation} \label{c1}
  \begin{split}
    L_\lambda f^2 L_\lambda & = D^* G_\lambda^{-1} D f^2 D^* G_\lambda^{-1} D + 2 D^* G_\lambda^{-1} f^2 U D + (Uf)^2 + U_1 \\
    & \quad + D^* G_\lambda^{-1} D f^2 T + T f^2 D^* G_\lambda^{-1} D + T f^2 U + U f^2 T + T f^2 T
  \end{split}
\end{equation}
where
\[
  \begin{split}
    U_1 & \coloneqq D^* G_\lambda^{-1} [D, f^2 U] + [Uf^2, D^*] G_\lambda^{-1} D \\
    & = \{ D^* G_\lambda^{-1} \{ Df^2U \} \} = \{ D_x^*(G_\Sigma+C_\lambda)^{-1} \{D_x f^2 U\} \} + \lambda^2 \{D_y^* D_y^{} f^2 U \}.
  \end{split}
\]
Calculating $\langle \psi, L_\lambda f^2 L_\lambda \psi \rangle$, and using the Cauchy-Schwarz inequality to get a lower bound for the second line in \eqref{c1}, we obtain
\begin{equation} \label{c2}
  \begin{split}
    L_\lambda f^2 L_\lambda & \ge 2^{-1} D^* G_\lambda^{-1} D f^2 D^* G_\lambda^{-1} D + 2 D^* G_\lambda^{-1} f^2 U D + 2^{-1}(Uf)^2 \\
    & \quad - C T f^2 T - |U_1|.
  \end{split}
\end{equation}
Similarly we prove that
\begin{equation} \label{c3}
  C(1 + \lambda^{-2} L_\lambda) f^2 (1 + \lambda^{-2} L_\lambda) \geq \lambda^{-4} L_\lambda f^2 L_\lambda.
\end{equation}
We next estimate $T f^2 T$ and $|U_1|$.

Using the Cauchy-Schwarz inequality and calculating we find
\begin{equation} \label{c4}
  \begin{split}
    T f^2 T & \le C (fD^* G_\lambda^{-1} A_\lambda^{})^* (f D^* G_\lambda^{-1}) + C(f A_\lambda^* G_\lambda^{-1} D)^* (f A_\lambda^* G_\lambda^{-1} D) \\
    & \le C (f A_\lambda^* G_\lambda^{-1} D)^* (f A_\lambda^* G_\lambda^{-1} D) + C \{ f D^* G_\lambda^{-1} A_\lambda^{} \}^2 \\
    & \le C D_x^* f^2 D_x^{} + C D_y^* \lambda^2 f^2 D_y^{} + C \\
    & \le \lambda^2 C(1 + \lambda^{-2} L_\lambda).
  \end{split}
\end{equation}
Here we used part \eqref{l:e1} and the estimate
\[
  \{ f D^* G_\lambda^{-1} A_\lambda^{} \} \le C + \lambda \{f D_y A_V(x,y/\lambda)\} \le C,
\]
which follows by a short calculation using the error bounds in \eqref{Ej}.

We now indicate how to prove that
\begin{equation} \label{c5}
  |U_1| \le C \tilde{\chi}^2 \tilde{\eta}_1^2 \lambda^4 ( 1 + \lambda^2 W_\lambda ),
\end{equation}
where $\tilde{\chi}$ and $\tilde{\eta}_1$ are like $\chi$ and $\eta_1$, but they are supported on a slightly larger subset. (There exist such $\tilde{\chi}$ and $\tilde{\eta}_1$ because of the properties of $\chi$ and $\eta_1$.) To obtain this estimate we write $U = E_0 + \lambda^2(2^{-1} wy^2 + \lambda^2 E_3)$, calculate $U_1$ using the product rule $D^TMv = \{D^TM\}v + (MD)^Tv$, where $M$ is a matrix and $v$ is a vector, and then use \eqref{Ej} and Lemma \ref{l:exp}. The first and second derivatives of $\chi$ and $\eta_1$ are bounded by a constant times $\tilde{\chi}$ and $\tilde{\eta}_1$, respectively. This gives the desired estimate.

Therefore, by combining \eqref{c5}, \eqref{c4} and \eqref{c3} into \eqref{c2}, by multiplying the inequality by $R_\lambda$ on both sides and rewriting it, and by using part \eqref{l:e1} to estimate the term which arises from \eqref{c5}, we obtain
\[
  \lambda^{-2} \| f D^* G_\lambda^{-1} D R_\lambda \| + \lambda^{-2} \| f U^{1/2} G_\lambda^{-1/2} D R_\lambda \| + \lambda^{-2} \| f U R_\lambda \| \le C.
\]
We next derive consequences from this estimate. From $\lambda^{-2} \| f U R_\lambda \| \le C$ it follows that $\| f(1 + \lambda^2 W_\lambda) R_\lambda \| \le C$. This proves \eqref{l:e2} for $l \le 2$ with $|\alpha| + p = 0$. From $\lambda^{-2} \| f U^{1/2} G_\lambda^{-1/2} D R_\lambda \| \le C$ we obtain $\| f \lambda W_\lambda^{1/2} \lambda^{-1} D_x R_\lambda \| \le C$ and $\| f \lambda W_\lambda^{1/2} D_y R_\lambda \| \le C$. This combined with part \eqref{l:e1} proves \eqref{l:e2} for $l=1$ with $|\alpha|+p=1$. The estimate in part \eqref{l:e2} for $l=0$ with $|\alpha|+p=1$ follows from part \eqref{l:e1}. The $l=0$ case with $|\alpha|+p=2$ is a consequence of $\lambda^{-2} \| f D^* G_\lambda^{-1} D R_\lambda \| \le C$. This is equivalent to $\lambda^{-2} \| D^* G_\lambda^{-1} D f R_\lambda \| \le C$ because the commutator terms are bounded, by part \eqref{l:e1}. Now observe that
\[
  \begin{split}
    \lambda^{-4} \| D^* G_\lambda^{-1} D f R_\lambda \psi \|^2 & \ge 2^{-1} \lambda^{-4} \| D^T G_\lambda^{-1} D f R_\lambda \psi \|^2 - C \\
    & \ge C \| [\lambda^{-1} D_x \ D_y]^T [\lambda^{-1} D_x \ D_y] f R_\lambda \psi \|^2 - C \\
    & \ge C \| (\lambda^{-1} D_x)^\alpha D_y^p f R_\lambda \psi \| - C.
  \end{split}
\]
Here, the first inequality follows by Cauchy-Schwarz inequality and part \eqref{l:e1}, the second inequality (uniform ellipticity) holds because $(G_\Sigma+C_\lambda)^{-1}$ and $1$ are symmetric and positive-definite (recall the block form of $G_\lambda^{-1}$), and the third inequality follows by Cauchy-Schwarz inequality. Thus, by moving $f$ back to the left and again using part \eqref{l:e1} to estimate the commutator terms, we obtain
\[
  \| f (\lambda^{-1} D_x)^\alpha D_y^p R_\lambda \| \le C.
\]
This completes the proof of part \eqref{l:e2}.

We finally prove part \eqref{l:e3}. Write $f = f_1^l f$, where $f_1 = \chi_1 \tilde{\eta}_1$, with $\chi_1$ such that $\chi_1 \chi = \chi$ and $\tilde{\eta}_1$ as in the proof of part \eqref{l:e2}. We begin by showing that
\begin{equation} \label{d1}
  \| f (1 + \lambda^2 W_\lambda)^{l/2} R_\lambda^{l+1} \| \le C.
\end{equation}
Set $g = f_1 (1 + \lambda^2 W_\lambda)^{1/2}$ and write \eqref{Llame} as
\[
  L_\lambda = D_x^* E_0^{} D_x + D_x^* E_0^{} + E_0^* D_x^{} + E_0 + \lambda^2 D_y^* D_y^{} + \lambda^4 W_\lambda.
\]
Then calculating we find
\[
  g^l R_\lambda^l = g R_\lambda g^{l-1} R_\lambda^{l-1} + g R_\lambda [\lambda^{-2} L_\lambda, g^{l-1}] R_\lambda^l
\]
and
\begin{equation} \label{d2}
  \begin{split}
    [\lambda^{-2} L_\lambda, g^{l-1}] & = \lambda^{-1} D_x^* E_0 \{D_x \lambda^{-1} g^{l-1} \} + \lambda^{-1} E_0^* \{ D_x \lambda^{-1} g^{l-1} \} \\
    & \quad + E_0 \{ D_x^{} D_x^T \lambda^{-2} g^{l-1} \} + D_y^* 2\{D_y g^{l-1} \} - \{D_y D_y g^{l-1} \} \\
    & = ( \lambda^{-1} D_x^* J_1 + D_y^* J_2 + J_3 ) (1+\lambda^2 W_\lambda)^{(l-2)/2} \\
    & = (1+\lambda^2 W_\lambda)^{(l-2)/2} ( \tilde{J}_1 \lambda^{-1} D_x + \tilde{J}_2 D_y + \tilde{J}_3 ),
  \end{split}
\end{equation}
where $J_k$ and $\tilde{J}_k$ (for $k=1,\dots,3$) are bounded functions with support contained in the support of $f_1$. To prove this, it is crucial to note that
\[
  \frac{1}{\lambda}(1+\lambda^2 W_\lambda)^{1/2}, \ \frac{\{D_x \lambda^2 W_\lambda \}}{\lambda (1+\lambda^2 W_\lambda)^{1/2}}, \ \frac{\{ D_x D_x^T \lambda^2 W_\lambda \}}{\lambda^2 (1+\lambda^2 W_\lambda)^{1/2}}, \ \text{etc.}
\]
multiplied by $f_1$ (or derivatives of it) are bounded functions. This follows by using \eqref{l:exp2} and the error bounds in \eqref{Ej} (note that $\tilde{\eta}_1(y) = \phi(y/\lambda)$ for $\phi \in C_0^\infty(\R)$). Thus, by Lemma \ref{l:e}\eqref{l:e1},
\[
  \| g^l R_\lambda^l \| \le C \|g^{l-1} R_\lambda^{l-1} \| + C \| f_2 (1 + \lambda^2 W_\lambda)^{(l-2)/2} R_\lambda^{l-1} \|,
\]
where ${f}_2$ is like $f_1$, but is supported on a slightly larger subset. Thus, by induction, we obtain \eqref{d1} (note that $\| g^l R_\lambda^{l+1} \| \le C \| g^l R_\lambda^l \|$).

We now set $T_{\alpha,p} = (\lambda^{-1} D_x)^\alpha D_y^p$ and write $T = T_{\alpha,p}$ with $|\alpha| + p \le 2$. Then
\[
  \begin{split}
    \| g^l T R_\lambda^{l+1} \| & = \| [g^l, T] R_\lambda^{l+1} + T f_1 g^l R_\lambda^{l+1} \| \\
    & = \| [g^l, T] R_\lambda^{l+1} + T f_1 R_\lambda g^l R_\lambda^l + T f_1 R_\lambda [\lambda^{-2} L_\lambda, g^l] R_\lambda^{l+1} \| \\
    & \le \| [g^l, T] R_\lambda^{l+1} \| + \| T f_1 R_\lambda \| ( \| g^l R_\lambda^l \| + \| [\lambda^{-2} L_\lambda, g^l] R_\lambda^{l+1} \|).
  \end{split}
\]
Furthermore
\[
  \| [g^l, T] R_\lambda^{l+1} \| \le \sum_{|\beta|+k \le 1} \| J_{l-1,\beta,k} T_{\beta,k} R_\lambda^l \|,
\]
and, using the last line in \eqref{d2}, we obtain
\[
  \| [\lambda^{-2} L_\lambda, g^l] R_\lambda^{l+1} \| \le \sum_{|\beta|+k \le 1} \| \tilde{J}_{l-1,\beta,k} T_{\beta,k} R_\lambda^l \|,
\]
with both $|J_{l-1,\beta,k}|$ and $|\tilde{J}_{l-1,\beta,k}|$ bounded above by $C f_3 (1 + \lambda^2 W_\lambda)^{(l-1)/2}$, where $f_3$ is like $f_2$, but is supported on a slightly larger subset. Therefore, by induction, we obtain the inequality of part \eqref{l:e3}.
\end{proof}

\section{Propagation bounds} \label{s:pbounds}

It follows from Lemma \ref{l:e} that the expected values of $(1+\lambda^2 W_\lambda)^{1/2}$ and $D_y$ are bounded by a constant, but it follows only that the expected value of $D_x$ is bounded by a constant times $\lambda$. We next improve this estimate by showing that, in the time interval $[0,T]$, the expected value of $D_x$ is actually bounded by a constant, and thus the energy transfered from normal motions is finite. This is the content of Lemma \ref{l:Dx} below.

We will use an energy cutoff. Let $\eta$ be a smooth function on $\R$ supported in $(0,\mu)$ with $\mu > 0$ that equals $1$ on $(0,\mu/2)$ with $0 \le \eta \le 1$. We define
\begin{equation} \label{etaE}
  \eta_{\sharp E} = \eta(\lambda^{-2} L_{\sharp \lambda}).
\end{equation}
We will also need a cutoff with a slightly larger support. We thus define $\tilde{\eta}_{\sharp E} = \tilde{\eta}(\lambda^{-2} L_{\sharp \lambda})$, where $\tilde{\eta}$ is a smooth function on $\R$ supported in $(0,\mu)$ with $\tilde{\eta} \eta = \eta$ and $0 \le \tilde{\eta} \le 1$.

Our goal in this section is to prove the following lemma:

\begin{lem}[Propagation bounds for $D_x$] \label{l:Dx}
  Assume the hypotheses and use the notation of Lemma \ref{l:e}. Let $\psi \in C_0^\infty(N\Sigma)$. Then for any $T < \infty$, there exist constants $\lambda_0$ and $C$ such that, for every $\lambda > \lambda_0$,
  \[
    \sup_{t \in [0,T]} \| \eta_1 D_x \chi e^{-itL_{\sharp \lambda}} \eta_{\sharp E} \psi \| \le C.
  \]
\end{lem}

To prove this lemma we will use Lemma \ref{l:e} and some additional estimates, which we state below in Lemmas \ref{l:e2+} and \ref{l:dyeta}. First, we need some definitions.

Using the partition of unity $\{\chi_j\}$ given in \eqref{pu}, we define
\[
  B = \sum_{j=1}^m \chi_j D_x^* G_\Sigma^{-1} D_x \chi_j.
\]
Here, in each term in the summation, $x$ is the local coordinate for the chart in which $\chi_j$ is supported. We then use the cutoff $\eta_1^{}$ given in \eqref{eta1} and set
\[
  B_+ = \eta_1 B \eta_1 + 1.
\]
In this definition, we included the identity to re-gain positivity so that $B_+^{-1/2}$ is well-defined (note that $\eta_1 B \eta_1$ is only nonnegative). Both operators $B$ and $B_+$ are essentially self-adjoint on $C_0^\infty(N\Sigma)$. (We may prove this by following the argument in the proof of \cite[Theorem 5.2]{S}.)

The additional estimates to proving Lemma \ref{l:Dx} are the following.

\begin{lem}[In addition to Lemma \ref{l:e}\eqref{l:e2}] \label{l:e2+}
  Assume the hypotheses and use the notation of Lemma \ref{l:e}. Then there exist constants $\lambda_0$ and $C$ such that, for every $\lambda > \lambda_0$, $l + |\alpha| + p \le 2$, and $|\beta| = 1$, we have
  \[
    \| (1 + \lambda^2 W_\lambda^{})^{l/2} (\lambda^{-1} D_x^{})^\alpha D_y^p D_x^\beta \chi \eta_1^{} \eta^{}_{\sharp E} B_+^{-1/2} \| \le C.
  \]
\end{lem}

\begin{lem} \label{l:dyeta}
  Assume the hypotheses and use the notation of Lemma \ref{l:e}. Let $\eta = \eta(n)$ be a smooth function on $N \Sigma$ supported in $N \Sigma_{\lambda^s \delta}$ that equals $1$ on $N \Sigma_{\lambda^s \delta_0}$ with $0 < s \le 1$, $0 < \delta_0 < \delta$, and $0 \le \eta \le 1$. Then there exist constants $\lambda_0$ and $C$ such that, for every $\lambda > \lambda_0$, $l > 0$, and $|\alpha| \le 1$, we have
  \[
    \| \chi \{D_y \eta\} (\lambda^{-1} D_x)^\alpha D_y R_{\sharp \lambda}^{l+1} \| \le C \lambda^{-sl}.
  \]
\end{lem}

Before we prove lemmas \ref{l:e2+} and \ref{l:dyeta}, let us use them to prove Lemma \ref{l:Dx}.

\begin{proof}[Proof of Lemma \ref{l:Dx}]
See Remark \ref{rem}. Set
\[
  b(t) = \langle e^{-itL_\lambda} \eta_E \psi, B_+ e^{-itL_\lambda} \eta_E \psi \rangle.
\]
Note that, since $\eta_1 B \eta_1 = B_+ - 1$, the lemma is proved if we show that
\[
  \sup_{t \in [0,T]} b(t) \le C.
\]
This estimate follows by Gronwall's inequality if we show that
\begin{equation} \label{gw}
  \frac{db}{dt} \le C b \qquad \text{and} \qquad b(0) \le C.
\end{equation}
Therefore, our goal is to prove \eqref{gw}.

We begin by showing that $b(0) \le C$. In fact
\[
  b(0) = \langle \eta_E \psi, B_+ \eta_E \psi \rangle = \| \eta_E \psi \|^2 + \sum_{j=1}^m \| G_\Sigma^{-1/2} D_x \chi_j \eta_1 \eta_E \psi \| \le C + mC \le C
\]
since
\[
  \| G_\Sigma^{-1/2} D_x \chi_j \eta_1 \eta_E \psi \| \le C \sum_{|\beta|=1} \| D_x^\beta \chi_j \eta_1 \eta_E B^{-1/2} \| \| B_+^{1/2} \psi \| \le C
\]
by Lemma \ref{l:e2+} (with $\alpha=0$, $p=0$, and $l=0$), and because $\| B_+^{1/2} \psi \| \le C$ for $\psi \in C_0^\infty(N\Sigma)$.

We now turn to the proof of $db/dt \le C b$. Observe that
\[
  \frac{d}{dt}b(t) = \langle \eta_E e^{-itL_\lambda} \psi, \tilde{\eta}_E [iL_\lambda, B_+] \tilde{\eta}_E \eta_E e^{-itL_\lambda} \psi \rangle \le C b(t)
\]
provided $\tilde{\eta}_E [iL_\lambda, B_+] \tilde{\eta}_E \le C B_+$, that is, provided
\begin{equation} \label{goal2}
  \sum_{j=1}^m \tilde{\eta}_E [iL_\lambda, \eta_1 \chi_j D_x^* G_\Sigma^{-1} D_x \chi_j \eta_1] \tilde{\eta}_E \le C B_+,
\end{equation}
where $\tilde{\eta}_E$ is the cutoff operator defined next to \eqref{etaE}. Thus, we are left to prove \eqref{goal2}. For this purpose, it suffices to write \eqref{Llame} as
\[
  \begin{split}
    L_\lambda & = D_x^* (G_\Sigma^{-1} + E_1^{}) D_x^{} + D_x^* (G_\Sigma^{-1} E_0^{}(x) + E_1^{}) + (E_1^* + E_0^*(x) G_\Sigma^{-1}) D_x^{} \\
    & \quad + E_0(x) + E_1 + \lambda^2(D_y^* D_y^{} + \lambda^2 W_\lambda) \\
    & \eqqcolon L_{1, \lambda} + \lambda^2 D_y^* D_y^{}.
  \end{split}
\]
We next estimate the contribution arising from $L_{1,\lambda}$ and $\lambda^2 D_y^* D_y^{}$ to \eqref{goal2}.

Let us prove that the contribution arising from $\lambda^2 D_y^* D_y^{}$ to the left-hand side of \eqref{goal2} is bounded by a constant. A simple calculation shows that
\begin{multline*}
  [ \lambda^2 D_y^* D_y^{}, \eta_1 \chi_j D_x^* G_\Sigma^{-1} D_x \chi_j \eta_1 ] \\
  = \lambda^2 ( \chi_j ( D_y^* \{D_y \eta_1\} + \, \text{adjoint} ) D^* G_\Sigma D_x \chi_j \eta_1 + \, \text{adjoint}).
\end{multline*}
Thus, the following estimate suffices to prove the claim:
\[
  \begin{split}
    & \lambda^2 |\langle \psi, \tilde{\eta}_E \chi_j D_y^* \{D_y \eta_1\} D_x^* G_\Sigma^{-1} D_x \chi_j \eta_1 \tilde{\eta}_E \psi \rangle| \\
    & \le \lambda^2 \| G_\Sigma^{-1/2} D_x \{D_y \eta_1\} D_y \chi_j \tilde{\eta}_E \psi \| \|G_\Sigma^{-1/2} D_x \chi_j \eta_1 \tilde{\eta}_E \psi \| \\
    & \le C \big( \lambda^4 \| \chi_j \{D_y \eta_1\} \lambda^{-1} D_x D_y R_\lambda^5 \| + \lambda^3 \| \{ D_x \chi_j \} \{ D_y \eta_1 \} D_y R_\lambda^5 \| \big) \\
    & \quad \times \| R_\lambda^{-5} \tilde{\eta}_E \| \| \eta_1 \lambda^{-1} D_x \chi_j R_\lambda^{1/2} \| \| R_\lambda^{-1/2} \tilde{\eta}_E \| \le C.
  \end{split}
\]
This follows by Lemma \ref{l:e}\eqref{l:e1} and Lemma \ref{l:dyeta} (with $s=1$ and $l=4$).

We now estimate the contribution arising from $L_{1,\lambda}$ to the left-hand side of \eqref{goal2}. We perform the calculations in a way that is suitable to make use of cancellations. Define $h_j = D_x^* G_\Sigma^{-1} D_x^{}$, $m_j = - \{D_x^T \chi_j\} G_\Sigma^{-1} \{D_x^{} \chi_j\}$, and $\mathcal{M} = \sum_{j=1}^m m_j$. Then
\[
  \chi_j D_x^* G_\Sigma^{-1} D_x^{} \chi_j = \tfrac{1}{2} \chi_j^2 h_j^{} + \tfrac{1}{2} h_j^{} \chi_j^2 + m_j.
\]
All these expressions are written in the chart $(\mathcal{V}_j, (x,y))$. Hence
\begin{align}
  & \sum_{j=1}^m [ iL_{1,\lambda}, \chi_j D_x^* G_\Sigma^{-1} D_x^{} \chi_j ] \notag \\
  & = \sum_{k=1}^m \sum_{j=1}^m \tfrac{1}{2} [i L_{1,\lambda}, \chi_j^2](h_j - h_k) \chi_k^2 + \tfrac{1}{2} \chi_k^2 (h_j - h_k) [i L_{\lambda,1}, \chi_j^2] \label{l1} \\
  & \quad + [iL_{1,\lambda}, \mathcal{M}] + \sum_{j=1}^m \tfrac{1}{2} [i L_{1,\lambda}, h_j] \chi_j^2 + \tfrac{1}{2} \chi_j^2 [i L_{1,\lambda}, h_j], \label{l2}
\end{align}
where we used that $\sum_{k=1}^m \chi_k^2=1$. Next, we want to express each term in \eqref{l1} in the chart $(\mathcal{V}_k, (x,y))$; we thus relabel $(\mathcal{V}_j, (x,y))$ as $(\mathcal{V}_j, (\tilde{x},\tilde{y}))$, and note that
\[
  D_{\tilde{x}} = M^T D_x + y N D_y \qquad \text{with} \qquad M G_\Sigma^{-1}(\tilde{x}) M^T = G_\Sigma^{-1}(x),
\]
where $M=M(x)$ is a matrix and $N=N(x)$ is a vector which arise from the Jacobian of $(\tilde{x},\tilde{y}) \mapsto (x,y)$ (we may think of $M$ and $N$ as functions of $x$). Thus
\[
  h_j - h_k = (y E_0^* D_y + E_0^*) D_x + y^2 E_0 D_y D_y + y E_0 D_y.
\]
Here, $E_0$ is either a vector or a scalar which depend only on $x$. Furthermore
\[
  [L_{1,\lambda}, \chi_j^2] = \chi_j D_x^* E_0 + \tilde{\chi}_j E_0,
\]
where $\tilde{\chi}_j = \tilde{\chi}_j(\sigma)$ is a smooth function on $N\Sigma$ supported in $\mathcal{V}_j$. Therefore
\begin{equation} \label{l3}
  \begin{split}
    \eqref{l1} = & \sum_{k=1}^m \big( \chi_k D_x^* (y E_0 D_y + (1+y) E_0) D_x \chi_k \\
    & \quad + \chi_k D_x^* (y^2 E_0 D_y D_y + (y+y^2) E_0 D_y + (1+y+y^2) E_0) \\
    & \quad + \tilde{\chi}_k (y^2 E_0 D_y D_y + (y+y^2) E_0 D_y + (1+y+y^2) E_0) \big).
  \end{split}
\end{equation}
Recall the presence of $\tilde{\eta}_E \eta_1$ around \eqref{l1} in \eqref{goal2}, and the estimate for $y$ in \eqref{l:exp2}. Then, by Lemma \ref{l:e2+} (with $\alpha=0$ and $(l,p)$ equal to $(0,0)$, $(1,0)$ and $(1,1)$), the first line in \eqref{l3} gives a contribution to \eqref{goal2} which is bounded by $C B_+$. By Lemma \ref{l:e2+} (with $\alpha=0$ and $(l,p)=(0,0)$), and Lemma \ref{l:e}\eqref{l:e2} (with $\alpha=0$ and $p \le 2$), the second and third line in \eqref{l3} give a contribution to \eqref{goal2} which is bounded by $C B_+^{1/2}$ and $C$, respectively. Thus, we are left to estimate the contribution arising from \eqref{l2} to \eqref{goal2}.

We next consider the quantity in \eqref{l2}. We can re-expand $\mathcal{M} = \mathcal{M}(x)$ as $\mathcal{M} = \sum_{k=1}^m \mathcal{M} \chi_k^2$, and then we find (using re-expansion again)
\[
  [L_{1,\lambda}, \mathcal{M}] = \sum_{k=1}^m \chi_k (D_x^* E_0^{} + E_0^{}).
\]
By Lemma \ref{l:e2+}, similarly as above, this gives a contribution to \eqref{goal2} which is bounded by $CB_+^{1/2}$. Finally, we estimate the summation in \eqref{l2}. This can be written (after some calculation) as
\[
  \begin{split}
    \sum_{j=1}^m \big( & \chi_j D_x^* (E_1^* D_x + E_1 + E_0) D_x \chi_j + \chi_j D_x^* E_0 + \tilde{\chi}_j (E_0^* D_x + E_0) \\
    & + \tilde{\chi}_j D_x^* E_0 \{ D_x \lambda^4 W_\lambda \} \chi_j + \chi_j \{D_x^T \lambda^4 W_\lambda \} \tilde{\chi}_j (E_0 D_x + E_0) \big),
  \end{split}
\]
where $\tilde{\chi}_j \chi_j = \chi_j$. Again, by Lemma \ref{l:e2+}, the first line in this expression gives a contribution to \eqref{goal2} which is bounded by $C B_+$. Recall from \eqref{l:exp2} that $\{ D_x \lambda^4 W_\lambda \} = O(y^4)$. Thus, by Lemmas \ref{l:e2+} and \ref{l:e}\eqref{l:e3} (with $l=4$ and $|\alpha|+p=0$), the second line in this expression gives a contribution to \eqref{goal2} which is bounded by $C B_+^{1/2}$. This completes the proof of \eqref{goal2} and proves the lemma.
\end{proof}

Finally, we prove Lemmas \ref{l:e2+} and \ref{l:dyeta}.

\begin{proof}[Proof of Lemma \ref{l:e2+}]
See Remark \ref{rem}. Let $\eta_{1,1} = \eta_{1,1}(n)$ be a smooth function on $N \Sigma$ supported in $N \Sigma_{\lambda \delta}$ with $\eta_{1,1} \eta_1 = \eta_1$ and $0 \le \eta_1 \le 1$. Let $\chi_1 = \chi_1(\sigma)$ be a smooth function on $N\Sigma$ supported in $\mathcal{V}$ with $\chi_1 \chi = \chi$ and $0 \le \chi_1 \le 1$. Set $f = \chi \eta_1$ and $F = (1 + \lambda^2 W_\lambda)^{l/2} (\lambda^{-1} D_x)^\alpha D_y^p \chi_1$. Then
\begin{equation} \label{b1}
  F D_x^\beta f \eta_E^{} B_+^{-1/2} = F \eta_{1,1}^{} \eta_E^{} D_x^\beta f B_+^{-1/2} + F \eta_{1,1}^{} [D_x^\beta f, \eta_E^{}] B_+^{-1/2}.
\end{equation}
Observe that, on the left-hand side of this equality, the term that arises from moving $\chi_1$ to the right is bounded:
\begin{multline*}
  \| (1+\lambda^2 W_\lambda^{})^{l/2} (\lambda^{-1} D_x^{})^\alpha D_y^p \{D_x^\beta \chi_1^{}\} \chi \eta_1^{} \eta_E^{} B_+^{-1/2} \| \\
  \le \| (1+\lambda^2 W_\lambda^{})^{l/2} (\lambda^{-1} D_x^{})^\alpha D_y^p \chi \eta_1^{} R_\lambda \| \| R_\lambda^{-1} \{D_x^\beta \chi_1^{} \} \eta_E^{} \| \| B_+^{-1/2} \| \le C.
\end{multline*}
This follows by Lemma \ref{l:e}(\ref{l:e1}-\ref{l:e2}). Thus, to prove the lemma it suffices to estimate the right-hand side of \eqref{b1}. As we will shortly explain, we have
\[
  \| F \eta_{1,1}^{} \eta_E^{} D_x^\beta f B_+^{-1/2} \| \le \| F \eta_{1,1}^{} R_\lambda^{} \| \| R_\lambda^{-1} \eta_E^{} \| \| D_x^\beta f B_+^{-1/2} \| \le C.
\]
In this inequality, the first factor is bounded, again by Lemma \ref{l:e}(\ref{l:e1}-\ref{l:e2}). The second factor is clearly bounded by $(1+\mu)$. The third factor is bounded because $f (D_x^\beta)^* D_x^\beta f \le C B_+$. Thus, to complete the proof of the lemma we need to estimate the right-hand side of the inequality
\begin{equation} \label{b2}
  \begin{split}
    \| F \eta_{1,1} [D_x^\beta f, \eta_E] B_+^{-1/2} \| & = \| F \eta_{1,1} R_\lambda^{} R_\lambda^{-1} [D_x^\beta f, \eta_E^{}] R_\lambda^{} \tilde{\eta}_E^{} R_\lambda^{-1} \| \\
    & \le C \| R_\lambda^{-1} [D_x^\beta f, \eta_E^{}] R_\lambda \|.
  \end{split}
\end{equation}
Here, $\tilde{\eta}_E$ is the cutoff operator defined next to \eqref{etaE}.

Let us now estimate the right-hand side of \eqref{b2}. Since $\eta_E = \phi(R_\lambda)$ for $\phi \in C_0^\infty((0,2))$, by the Helffer-Sj\"ostrand formula, we have
\[
  \eta_E = \phi(R_\lambda) = \int_\C g_\phi(z) (R_\lambda - z)^{-1} dz \wedge d\bar{z}
\]
for an appropriated function $g_\phi$ \cite{D}. Using the Resolvent Identity, we find
\[
  R_\lambda^{-1} [D_x^\beta f, (R_\lambda^{} - z)^{-1}] R_\lambda^{} = (R_\lambda^{} - z)^{-1} [D_x^\beta f, \lambda^{-2} L_\lambda^{}] R_\lambda^{} (1+z(R_\lambda^{}-z)^{-1}).
\]
Thus
\begin{equation} \label{b3}
  \begin{split}
    & \| R_\lambda^{-1} [D_x^\beta f, \eta_E^{}] R_\lambda^{} \| \\
    & \le C \| [D_x^\beta f, \lambda^{-2} L_\lambda^{}] R_\lambda^{} \| \int_\C |g_\phi(z)| |\imag(z)|^{-1} (1 + |z| |\imag(z)|^{-1}) dz \wedge d\bar{z} \\
    & \le C \| [D_x^\beta f, \lambda^{-2} L_\lambda^{}] R_\lambda^{} \|.
  \end{split}
\end{equation}
The integral above is finite because of the properties of $g_\phi$ (see (5) and (H3) in \cite{D}). To estimate the right-hand side of \eqref{b3}, it suffices to write \eqref{Llame} as
\[
  L_\lambda^{} = D_x^* E_0^{} D_x^{} + D_x^* E_0^{} + E_0^* D_x^{} + E_0 + \lambda^2 D_y^* D_y^{} + \lambda^4 W_\lambda.
\]
Consequently
\begin{equation} \label{b5}
  \begin{split}
    [D_x^\beta f, \lambda^{-2} L_\lambda^{}] & = \eta_1^{} [D_x^\beta \chi, \lambda^{-2} D_x^* E_0^{} D_x^{} + \lambda^{-2} D_x^* E_0^{} + \lambda^{-2} E_0^* D_x^{} \\
    & \quad + \lambda^{-2}E_0^{}] + [D_x^\beta \chi \eta_1^{}, D_y^* D_y^{}] + \lambda^2 \eta_1 [D_x^\beta \chi, W_\lambda].
  \end{split}
\end{equation}
The first term on the right-hand side of this expression can be written as
\[
  \eta_1^{} \lambda^{-2} ( D_x^\beta \{D_x^T \chi\} E_0^{} D_x^{} + D_x^\beta D_x^* E_0^{} \{D_x^{} \chi\} + E_0^{} D_x^* E_0^{} D_x^{} \chi + D_x^* E_0^{} D_x^{} \chi),
\]
the second term can be written as
\[
  \eta_1^{} \lambda^{-2} ( E_0^{} \chi D_x^{} + E_0^* \{D_x^{} \chi\} D_x^\beta + E_0^* \{D_x^\beta D_x^{} \chi\} ),
\]
and we have a similar expression for the third term. The fourth term on the right-hand side of \eqref{b5} is given by
\[
  \lambda^{-2} \eta_1^{} \{ D_x^\beta \chi E_0 \}.
\]
All these terms give a bounded contribution to \eqref{b3} by Lemma \ref{l:e}(\ref{l:e1}-\ref{l:e2}) because they contain at most two derivative operators with respect to $x$ (multiplied by $\lambda^{-1}$). The fifth term on the right-hand side of \eqref{b5}, which is given by
\begin{multline*}
  - 2 \chi \lambda \{D_y \eta_1\} \lambda^{-1} D_x^\beta D_y^{} - \chi \lambda \{D_y D_y \eta_1\} \lambda^{-1} D_x^\beta \\
  - 2 \{D_x^\beta \chi\} \{D_y \eta_1\} D_y - \{D_x^\beta \chi\} \{D_y D_y \eta_1\},
\end{multline*}
gives a bounded contribution to \eqref{b3}. Again, this follows by Lemma \ref{l:e}\eqref{l:e2} and by noting that $\eta_1(y) = \phi(y/\lambda)$. Finally, the sixth term on the right-hand side of \eqref{b5} is equal to
\[
  \chi \eta_1 \{ D_x^\beta \lambda^2 W_\lambda \}.
\]
Recall from \eqref{l:exp2} that $\{D_x^\beta \lambda^2 W_\lambda \} = O(y^2)$. Thus, by Lemma \ref{l:e}\eqref{l:e2} (with $l=2$), the above term gives a bounded contribution to \eqref{b3}.

To summarize, we have shown that all the terms in \eqref{b5} give a bounded contribution to \eqref{b3} so that \eqref{b2} is bounded. This completes the proof of the lemma.
\end{proof}

Here is the proof of Lemma \ref{l:dyeta}:

\begin{proof}[Proof of Lemma \ref{l:dyeta}]
Let $\eta_1$ be the cutoff function in \eqref{eta1}, and let $\chi_1 = \chi_1(\sigma)$ be a smooth function on $N\Sigma$ supported in $\mathcal{V}$ with $\chi_1 \chi = \chi$ and $0 \le \chi_1 \le 1$. Observe that $\{ D_y \eta \}$ is supported in the region $\set{(\sigma,n)}{\lambda^s \delta_0 < |n| < \lambda^s \delta}$, and $\eta_1 = 1$ on the support of $\{D_y \eta\}$ for large $\lambda$ (namely $\lambda > (\delta/\delta_1)^{1/(1-s)}$). Thus, by Lemma \ref{l:e}\eqref{l:e3}, for any integer $l>0$,
\[
  \begin{split}
    & \| \chi \{D_y \eta\} (\lambda^{-1} D_x)^\alpha D_y R_{\sharp \lambda}^{l+1} \| \\
    & \qquad \le \lambda^{-sl} \| \chi_1 \{D_y \eta\} \lambda^{sl} \langle n \rangle^{-l} \| \| \chi \eta_1 \langle n \rangle^l (\lambda^{-1} D_x)^\alpha D_y R_{\sharp \lambda}^{l+1} \| \\
    & \qquad \le C \lambda^{-sl} \| \chi \eta_1 (1+\lambda^2 W_\lambda)^{l/2} (\lambda^{-1} D_x)^\alpha D_y R_{\sharp \lambda}^{l+1} \| \le C \lambda^{-sl}.
  \end{split}
\]
This is the desired estimate.
\end{proof}

\section{Proof of Theorem \ref{t:main}} \label{s:proof}

\begin{proof}[Proof of Theorem \ref{t:main}]
Since
\[
  \| e^{-itL_\lambda} \psi - e^{-itL_{0,\lambda}} \psi \|^2 = 2 \langle \psi, \psi \rangle - 2 \real \langle e^{-itL_\lambda} \psi, e^{-itL_{0,\lambda}} \psi \rangle,
\]
it is enough to prove that
\[
  \sup_{t \in [0,T]} |\langle e^{-itL_\lambda} \psi, e^{-itL_{0,\lambda}} \psi \rangle - \langle \psi, \psi \rangle| \le C \lambda^{s-1}
\]
for $0 < s < 1$ and $\psi$ in a dense subset of $L^2(N \Sigma, \dvol_{N \Sigma})$. Therefore, our goal is to prove this inequality for $\psi \in C^\infty_0(N \Sigma)$.

We begin the analysis by inserting the energy cutoff $\eta_{\sharp E}$ given in \eqref{etaE}. Since $\| L_{\sharp \lambda} \psi \| \le C \lambda^2$, we have
\[
  \| (1-\eta_{\sharp E}^{}) \psi \| \le \| (1-\eta_{\sharp E}^{}) L_{\sharp \lambda}^{-1} \| \| L_{\sharp \lambda}^{} \psi \| \le \mu^{-1} \lambda^{-2} C \lambda^2 \le C \mu^{-1}
\]
for any $\mu > 0$. Thus, it is enough to prove that, for each $\mu > 0$,
\[
  \sup_{t \in [0,T]} |\langle e^{-itL_\lambda} \eta_E^{} \psi, e^{-itL_{0,\lambda}} \eta_{0,E}^{} \psi \rangle - \langle \eta_E^{} \psi, \eta_{0,E}^{} \psi \rangle| \le C \lambda^{s-1}.
\]

We introduce also a stronger cutoff in the $n$ variable. Let $\eta_2 = \eta_2(n)$ be a smooth function on $N\Sigma$ supported in $N \Sigma_{\lambda^s/2}$ that equals $1$ on $N \Sigma_{\lambda^s/4}$ with $0 \le \eta_2 \le 1$. Define $\langle n \rangle = (1+|n|^2)^{1/2}$. Note that \eqref{l:exp2} implies that $\langle n \rangle \le C (1 + \lambda^2 W_\lambda)^{1/2}$. Thus, by Lemma \ref{l:e}\eqref{l:e1},
\[
  \| (1-\eta_2^{}) \eta_{\sharp E}^{} \| \le \lambda^{-s} \| (1-\eta_2^{}) \lambda^s \langle n \rangle^{-1} \| \| \langle n \rangle R_{\sharp \lambda}^{1/2} \| \| R_{\sharp \lambda}^{-1/2} \eta_{\sharp E}^{} \| \le C \lambda^{-s}.
\]
Therefore, to prove the theorem it suffices to show that
\begin{equation} \label{goal}
  \sup_{t \in [0,T]} |\langle e^{-itL_\lambda} \eta_E^{} \psi, \eta_2^{} e^{-itL_{0,\lambda}} \eta_{0,E}^{} \psi \rangle - \langle \eta_E^{} \psi, \eta_2^{} \eta_{0,E}^{} \psi \rangle| \le C \lambda^{s-1}
\end{equation}
for each $\mu > 0$ and every $\psi \in C_0^\infty(N \Sigma)$.

Let us prove \eqref{goal}. Recall from \eqref{Llame} that $L_\lambda = L_{0,\lambda} + Q_\lambda$. Using the partition of unity $\{\chi_j\}$ given in \eqref{pu}, and integrating the derivative, we write
\begin{multline} \label{a1}
  \langle e^{-itL_\lambda} \eta_E^{} \psi, \eta_2^{} e^{-itL_{0,\lambda}} \eta_{0,E}^{} \psi \rangle - \langle \eta_E^{} \psi, \eta_2^{} \eta_{0,E}^{} \psi \rangle \\
  = i \int_0^t \sum_{j=0}^m \langle e^{-isL_\lambda} \eta_E^{} \psi, \chi_j^{} ([L_{0, \lambda}^{}, \eta_2^{}] + \eta_2^{} \tilde{Q}_\lambda^{}) \chi_j^{} e^{-isL_{0,\lambda}} \eta_{0,E}^{} \psi \rangle ds,
\end{multline}
where (in local coordinates)
\[
  \tilde{Q}_\lambda = Q_\lambda + E_1 = (D_x + E_0)^* E_1 (D_x + E_0)
\]
and
\begin{equation} \label{a2}
  [L_{0,\lambda}^{}, \eta_2^{}] = \lambda^2 [D_y^* D_y^{}, \eta_2^{}] = \lambda^2 D_y^* \{ D_y^{} \eta_2^{} \} + \lambda^2 \{D_y^* \eta_2^{}\} D_y^{}.
\end{equation}
The term $E_1$ in $\tilde{Q}_\lambda$ arises from commuting $\chi_j$ with $Q_\lambda$ to the right in \eqref{a1} (note that $\chi_j$ commutes with $[L_{0,\lambda}, \eta_2]$). We next estimate the contribution arising from each of these terms to \eqref{a1}.

The contribution arising from $[L_{0, \lambda}, \eta_2]$ to the integrand of \eqref{a1} tends to zero faster than any inverse power of $\lambda$. We will prove this for the first term in \eqref{a2}; the proof for its adjoint is similar. In fact, by Lemma \ref{l:dyeta}, for any integer $l > 0$,
\[
  \begin{split}
    & \lambda^2 \langle \chi_j^{} \{D_y^{} \eta_2\} D_y^{} e^{-isL_\lambda} \eta_E^{} \psi, \chi_j^{} e^{-isL_{0, \lambda}} \eta_{0,E}^{} \psi \rangle \\
    & \le C \lambda^2 \| \chi_j^{} \{D_y \eta_2^{}\} D_y^{} \eta_E^{} e^{-isL_\lambda} \psi \| \le C \lambda^2 \| \chi_j^{} \{D_y \eta_2^{}\} D_y^{} R_\lambda^{l+1} \| \| R_\lambda^{-l-1} \eta_E \| \\
    & \le C (1+\mu)^{l+1} \lambda^{2-sl}.
  \end{split}
\]
By choosing $l$ as large as desired, this proves the claim.

The contribution arising from $\eta_2 E_1$ to the integrand of \eqref{a1} is bounded by
\[
  C \| \chi_j^{} \eta_2^{} E_1^{} \eta_{0,E}^{} \| \le C \lambda^{-1} \| \chi_j^{} \eta_2^{} \lambda E_1^{} \langle n \rangle^{-1} \| \| \langle n \rangle R_{0,\lambda}^{1/2} \| \| R_{0,\lambda}^{-1/2} \eta_{0,E}^{} \| \le C \lambda^{-1}.
\]
This follows by Lemma \ref{l:e}\eqref{l:e1}.

Finally we estimate the contribution arising from
\begin{equation} \label{a4}
  \eta_2^{} Q_\lambda^{} = D_x^* \eta_2^{} E_1^{} D_x^{} + D_x^* \eta_2^{} E_1^{} + E_1^* \eta_2^{} D_x^{}
\end{equation}
to the integrand of \eqref{a1}. First observe that
\[
  \| \eta_2^{} E_1^{} \| \le C \lambda^s/\lambda = C \lambda^{s-1}.
\]
Now, consider the first term in \eqref{a4}. Then, by Lemma \ref{l:Dx},
\begin{multline*}
  \langle e^{-isL_\lambda} \eta_E^{} \psi, \chi_j^{} D_x^* \eta_2^{} E_1^{} D_x^{} \chi_j^{} e^{-isL_{0,\lambda}} \eta_{0,E}^{} \psi \rangle \\
  \le \| \eta_1^{} D_x^{} \chi_j^{} e^{-isL_\lambda} \eta_E^{} \psi \| \| \eta_2^{} E_1^{} \| \| \eta_1^{} D_x^{} \chi_j^{} e^{-isL_{0,\lambda}} \eta_{0,E}^{} \psi \| \le C \lambda^{s-1}.
\end{multline*}
Similarly we obtain the same estimate for the other two terms in \eqref{a4}.

Therefore, by combining the estimates in the last three paragraphs into \eqref{a1}, we conclude that
\[
  \sup_{t \in [0,T]} |\langle e^{-itL_\lambda} \eta_E^{} \psi, \eta_2^{} e^{-itL_{0,\lambda}} \eta_{0,E}^{} \psi \rangle - \langle \eta_E^{} \psi, \eta_2^{} \eta_{0,E}^{} \psi \rangle| \le m T C \lambda^{s-1}.
\]
This proves \eqref{goal} and thus proves the theorem. \end{proof}

\end{document}